\documentclass[lettersize,journal]{IEEEtran}
\usepackage{amsmath,amsfonts}
\usepackage{algorithmic}
\usepackage{algorithm}
\usepackage{array}
\usepackage[caption=false,font=normalsize,labelfont=sf,textfont=sf]{subfig}
\usepackage{textcomp}
\usepackage{stfloats}
\usepackage{url}
\usepackage{verbatim}
\usepackage{graphicx}
\usepackage{cite}
\usepackage{bm}

\usepackage{mathtools}
\usepackage{amsthm}
\newtheorem{theorem}{Theorem}
\newtheorem{proposition}{Proposition}
\theoremstyle{definition}
\newtheorem{definition}{Definition}
\newtheorem{remark}{Remark}
\newtheorem{assumption}{Assumption}
\begin{document}
%%%%%%%%%%%%%%%%%%%%%%%%%%%%%%%%%%%%%%%%%%%%%%%%%%%%%%%%%%%%%%%%%%%%%%%%%%%%%%%%%%%%%%%%%%%%%%%%%%%%%%%%%%%%%%%%%%%%%%%%%%%%%%%%%%%%%%

\title{Characterization and Analysis of Emergency Landing Flight Envelopes with Graded Safety Specifications}

% \title{Emergency Landing Flight Envelope Analysis\\with a Graded Notion of Safety}

% \author{Chams~E.~Mballo,
%         Bryce~L.~Ferguson,
%         Inkyu~Jang,
%         Donggun~Lee,
%         and Claire~J.~Tomlin%
% \thanks{C. E. Mballo and C. J. Tomlin are with the Department of Electrical Engineering and Computer Sciences, University of California, Berkeley, CA 94720 USA.}
% \thanks{B. L. Ferguson is with the Thayer School of Engineering, Dartmouth College, Hanover, NH 03755 USA.}
% \thanks{I. Jang is with the Department of Aerospace Engineering, Seoul National University, Seoul 08826, Republic of Korea.}
% \thanks{D. Lee is with the Department of Mechanical and Aerospace Engineering, North Carolina State University, Raleigh, NC 27695 USA.}
% \thanks{Corresponding author: C. E. Mballo (email: cmballo@berkeley.edu).}}

\author{Chams~Eddine~Mballo, Bryce~L.~Ferguson, Inkyu~Jang, Donggun~Lee,
and Claire~J.~Tomlin%
\thanks{Manuscript submitted May 25, 2026. The work of Chams Eddine Mballo was supported by the University of California President's Postdoctoral Fellowship Program.}%
\thanks{Chams Eddine Mballo is with the Department of Electrical Engineering and Computer Sciences, University of California, Berkeley, CA 94720 USA (e-mail: cmballo@berkeley.edu).}%
\thanks{Bryce L. Ferguson is with the Thayer School of Engineering, Dartmouth College, Hanover, NH 03755 USA (e-mail: bryce.l.ferguson@dartmouth.edu).}%
\thanks{Inkyu Jang is with the Department of Aerospace Engineering, Seoul National University, Seoul 08826, Republic of Korea (e-mail: leplusbon@snu.ac.kr).}%
\thanks{Donggun Lee is with the Department of Mechanical and Aerospace Engineering, North Carolina State University, Raleigh, NC 27695 USA (e-mail: donggun\_lee@ncsu.edu).}%
\thanks{Claire J. Tomlin is with the Department of Electrical Engineering and Computer Sciences, University of California, Berkeley, CA 94720 USA (e-mail: tomlin@eecs.berkeley.edu).}%
\thanks{Corresponding author: Chams Eddine Mballo.}%
}

%\thanks{Manuscript received Month Day, Year; revised Month Day, Year.}}

\markboth{}%
{Mballo \MakeLowercase{\textit{et al.}}: Emergency Landing Flight Envelope Computation under Soft Safety Constraints}

%%%%%%%%%%%%%%%%%%%%%%%%%%%%%%%%%%%%%%%%%%%%%%%%%%%%%%%%%%%%%%%%%%%%%%%%%%%%%%%%%%%%%%%%%%%%%%%%%%%%%%%%%%%%%%%%%%%%%%%%%%%%%%%%%%%%%%

% \title{Emergency Landing Flight Envelope Computation under Soft Safety Constraints}


% \markboth{IEEE TRANSACTIONS ON CONTROL SYSTEMS TECHNOLOGY,~2026}%

% IEEE TRANSACTIONS ON CONTROL SYSTEMS TECHNOLOGY, 2026

% % column for its text to clear the IEEEpubid mark.

\maketitle

% Emergency landing flight envelope analysis is traditionally based on a binary notion of safety, in which aircraft states are classified as either safe or unsafe. In practice, however, safety is inherently graded, spanning a continuum in the state space from the nominal regime, the aircraft’s intended mode of operation, through the degraded regime, an undesirable yet temporarily tolerable condition, and ultimately to the critical regime, a hazardous and operationally unacceptable condition. 

\begin{abstract}

Emergency landing flight envelope analysis traditionally adopts a binary notion of safety, whereby a trajectory is safe only if state constraints are satisfied pointwise in time. In practice, ensuring a successful landing requires recognizing that aircraft operation spans a continuum in the state space from the nominal to the critical regime. Between these regimes lies a degraded regime of states outside nominal operation that may be visited only for limited durations. Safety is therefore inherently graded, in the sense that limited exposure to degraded states may be tolerated, and must be assessed using a trajectory-dependent criterion rather than a purely pointwise-in-time one. This paper develops a Hamilton-Jacobi reachability framework for analyzing emergency landing flight envelopes under this graded notion of safety. Safety is encoded through a soft constraint defined by a designer-specified continuous violation cost function that assigns zero cost in the nominal regime and larger cost to more safety-critical off-nominal states. We introduce a general class of state- and time-dependent violation cost functions and establish monotonicity and continuity properties that characterize how the flight envelope varies with the cost of off-nominal operation. These results provide a principled sensitivity analysis linking safety conservativeness to operational capability. Building on this analysis, we propose a synthesis algorithm for parameterized violation cost functions in this class. The algorithm provably converges to the least conservative parameter under which a prescribed off-nominal safety requirement is satisfied. Numerical results for a fixed-wing emergency landing scenario under propulsion failure demonstrate the sensitivity properties and validate the algorithm.

\end{abstract}

\begin{IEEEkeywords}
Hamilton-Jacobi reachability, optimal control, safety verification, safety-critical control, soft constraints, emergency landing, flight envelope, aerospace systems.
\end{IEEEkeywords}

\section{Introduction}
\IEEEPARstart{S}{afety} in engineered control systems is often framed as the problem of characterizing a subset of the state space from which a system can be steered to its goal while avoiding unsafe states. Such characterizations typically account for uncertainty~\cite{c1, c2}. In aircraft operations, a prototypical example of this subset is the flight envelope~\cite{Tekles2017, Tang2009, Yavrucuk2009}, which specifies the set of states in which aircraft operation is {\em safe}, namely structural, aerodynamic, propulsion, and control authority limits are satisfied. In a propulsion failure scenario, one considers an emergency landing envelope, defined as the set of states from which a safe landing can still be achieved; see, e.g.,~\cite{Akametalu}. Determining this emergency landing envelope is naturally formulated as a robust reachability problem: one computes the set of initial conditions from which the aircraft can be guided to touchdown in the presence of disturbances while satisfying safety constraints.

In practice, flight envelopes are established primarily through flight testing~\cite{Walgemoed2005FlightEnvelope}. Through this process, limits on critical flight variables, commonly referred to as limit parameters~\cite{Unnikrishnan2011ReactionaryEnvelopeProtection}, are determined and used to characterize the flight envelope. With the advent of fly-by-wire systems, advanced flight control architectures have enabled carefree maneuvering through automatic flight envelope protection~\cite{Falkena2011,Yavrucuk2002AdaptiveLimitMarginCueing}. Complementary pilot cueing approaches have also been developed to alert or guide pilots near envelope limits, including tactile cues transmitted through the control inceptor, as investigated at NASA Ames~\cite{Whalley1994,WhalleyHindsonThiers2000}. For aircraft equipped with such technologies, showing compliance with certification requirements includes flying-qualities assessments to demonstrate that the added automation preserves task performance and maintains acceptable handling qualities without imposing excessive pilot workload~\cite{Lotterio2022AdvancedFlightControls,Klyde2020HQTE}.

Classical reachability formulations offer a complementary computational way to characterize such envelopes. These formulations typically adopt a binary notion of safety~\cite{c3,MargellosLygeros2011}: an initial state is classified as safe or unsafe according to whether there exists a feedback controller that can steer the aircraft from that state to touchdown while satisfying state constraints pointwise in time. In practice, however, safety in aircraft and many engineered systems is inherently graded, spanning a spectrum from \emph{nominal} (intended operating conditions), through \emph{degraded} (undesirable but tolerable for limited duration under exceptional circumstances), to \emph{critical} (poses an unacceptable hazard and is therefore operationally inadmissible). For example, during an emergency landing, the pilot may operate certain aircraft components slightly outside their nominal operating range, such as under higher structural and aerodynamic loads, to obtain the maneuvering capability required for a safe landing. Such operating conditions increase component wear and reduce safety margins without causing immediate component failure, and correspond to operating in the degraded region. This graded notion of safety highlights a fundamental trade-off in achieving a safe landing: prohibiting any excursions beyond the nominal region can be unsafe due to over-conservatism, whereas prolonged operation in degraded regimes increases the risk of system failure. Thus, achieving a safe emergency landing hinges on a careful balance between these two regimes of operation.

% This graded notion of safety highlights a fundamental trade-off in achieving a safe landing: prohibiting any excursions beyond the nominal region, especially in emergencies, can be unsafe due to over-conservatism, whereas prolonged operation in degraded regimes poses its own safety risks and can lead to system failure. 

The soft-constrained reachability formulation in~\cite{mballo} takes a first step toward systematically formalizing this graded notion of safety. In particular, it enforces staying within the nominal operating region as a soft safety constraint within the Hamilton-Jacobi (HJ) reachability framework. This is achieved by introducing an indicator-type violation cost function $\ell$ that measures the time the system spends in degraded regimes, together with a finite violation time budget that bounds the allowable cumulative time in these regimes. The resulting soft-constrained reach-avoid set encodes, for a given choice of $\ell$, the initial conditions from which a safe landing can be guaranteed while respecting both hard constraints (e.g., no operation in the critical region) and the prescribed budget on off-nominal operation. However, how $\ell$ should be designed for a given setting is not well understood, in part because the treatment in~\cite{mballo} is restricted to a simple indicator cost function. While appropriate for establishing a first-principles formulation of HJ reachability with soft safety constraints, this restriction precludes violation costs that depend on both the duration of off-nominal operation and its severity (i.e., the states visited during such excursions). In practice, the appropriate cost function should be informed by aircraft type and health condition. A richer class of violation cost functions, potentially expressed as a parameterized family $\ell_{\lambda}$, is therefore needed. A remaining open question is how the choice of $\ell$ (and, when parameterized, $\lambda$) influences the resulting flight envelope. As a result, there is little guidance on how to select $\ell$ in a way that is tied to crucial indicators of aircraft condition (e.g., structural health) and that it achieves a desired trade-off between operating in the nominal and degraded regimes. For instance, remaining component fatigue life can inform the conservativeness of $\ell$. An aircraft whose components are close to reaching their fatigue life calls for a more conservative $\ell$ that heavily penalizes operation in degraded regimes. By contrast, for an aircraft whose components have substantial remaining fatigue life, a less conservative $\ell$ is acceptable, since more sustained off-nominal operation can be tolerated.

% Although~\cite{mballo} primarily focuses on establishing the theoretical framework and illustrates one specific choice of $\ell$, it does not provide a systematic characterization of how the choice of $\ell$ and its parameters reshapes the emergency landing flight envelope. As a result, there is still little guidance on how to select or tune~$\ell$ in a way that is tied to crucial indicators of the aircraft’s condition (e.g., structural health) and that helps achieve a desired trade-off between operating in the nominal and degraded regimes. 

From the perspective of aircraft designers, operators and certification authorities, this leaves two key design questions largely unaddressed: \emph{(i)} how the emergency landing flight envelope changes as $\ell$ is varied across a spectrum of conservativeness, and \emph{(ii)} given a parameterized violation cost function, how can one tune its parameter to satisfy an explicit safety requirement on off-nominal operation while retaining as much operational capability (e.g., maximum safe landing altitude) as possible. The first question is fundamentally about sensitivity, whereas the second concerns synthesis, namely how to translate these sensitivity insights into a concrete design choice for~$\ell$.

To address these questions, this paper builds on the soft-constrained HJ reachability framework of~\cite{mballo} to compute the emergency landing flight envelope for a fixed-wing aircraft point-mass model under propulsion failure. We consider a one-parameter family of soft-constraint violation cost functions $\{\ell_\lambda\}_{\lambda \ge 0}$, without {\em a priori} specifying a particular functional form, to cover a broad class of violation cost functions. The parameter $\lambda$ controls the conservativeness of $\ell$. We address question \emph{(i)} above by characterizing how the envelope varies with $\lambda$, and \emph{(ii)} by leveraging this sensitivity analysis to derive a tuning procedure that selects $\lambda$ based on specified safety and performance requirements for the aircraft.

\subsection{Relevant Literature}
% Our methodology draws on insights from several research threads.

To situate our contribution, we briefly review two areas most relevant to this paper: flight envelope analysis and protection, including HJ reachability-based approaches, and soft-constrained methods in control and reinforcement learning.

Flight envelope analysis and protection for damaged or faulty aircraft is challenging because faults can alter aerodynamic coefficients and shift the center of gravity, leading to substantial changes in the aircraft dynamics that are difficult to predict {\em a priori}. Several works have used tools from optimal control, notably HJ reachability analysis, to determine operational flight envelopes and compute emergency landing flight envelopes under failure scenarios such as engine-out~\cite{Akametalu}, structural damage or failure~\cite{Lombaerts2013, Nabi2018}, partial loss of thrust~\cite{Schuet2017}, and actuator faults~\cite{Tang2009}. Other works have used HJ reachability analysis to synthesize envelope-protection logics that keep the aircraft inside the flight envelope, thereby preventing loss-of-control events~\cite{Schuet2017, Hsu}. Beyond flight envelope applications, HJ-based methods have also been used to compute reachable sets and associated guidance strategies that steer a fixed-wing aircraft with a failed engine to a safe landing site~\cite{Akametalu}, and to solve reach-avoid problems arising in the design of autolander systems for hybrid aircraft models~\cite{bayen}. HJ reachability has also been applied in air traffic management, where collision avoidance is formulated as a differential game and solved via backward reachable-set computation to provide safety guarantees for aircraft separation at high altitude~\cite{c3}. Outside aerospace, HJ reachability has seen widespread use in robotics and autonomous systems for motion planning and safety-critical control~\cite{Fisac}. Across these works, safety is treated as binary: states are classified as either safe or unsafe, without distinguishing among nominal, degraded, and critical regimes.

Optimization problems with soft constraints have been widely studied in control and reinforcement learning, particularly within frameworks such as Constrained Markov Decision Processes (CMDPs)~\cite{altman, Russel2020}, Model Predictive Control (MPC)~\cite{zelinger, wabersich}, and Control Barrier Functions (CBFs)~\cite{lee, xiao}. These methods introduce soft constraints to trade off safety and performance, typically by relaxing hard constraints and penalizing the resulting violations using slack variables or auxiliary cost terms~\cite{altman, Russel2020, zelinger, wabersich, lee, xiao, Chow2017, Tessler2019, mpc_soft_constraints1, mpc_soft_constraints2, mpc_soft_constraints5}. Although effective for synthesizing constraint-aware policies, these methods do not explicitly compute reachable sets such as reach-avoid sets. Consequently, they do not provide the set-based safety certificates afforded by HJ reachability, where the reachable set is given by a value function’s zero sublevel set. Moreover, these soft-constraint formulations have not been used to compute aircraft operational flight envelopes.

\subsection{Contributions of this paper}

We summarize the contributions of the paper as follows:

\begin{itemize}

% \item We introduce a general class of state- and time-dependent violation cost functions that extends the soft-constrained framework of~\cite{mballo}. This generalization enables a systematic sensitivity analysis over parameterized families from this class, characterizing how the emergency landing flight envelope changes as the cost assigned to degraded operation is varied.

\item We introduce a general class of state- and time-dependent violation cost functions that extends the soft-constrained framework of~\cite{mballo} beyond indicator-type violation costs. This generalization enlarges the admissible function space for encoding soft safety constraints, enabling the modeling of varying degrees of tolerance to off-nominal operation. It also enables a systematic sensitivity analysis over parameterized families of violation costs. This analysis characterizes how the emergency landing flight envelope varies as the cost assigned to degraded operation is tuned.

% \item We prove a general monotonicity result: if one violation cost function is more conservative than another, in the sense that it assigns no smaller cost to any state at any time, then the emergency landing flight envelope induced by the more conservative violation cost function is contained in that of the less conservative one. We also establish that the flight envelope set changes continuously with the cost function's parameters. Finally, we show that suitable parameterized families in this class recover, in the limit as the parameter tends to infinity, either the soft-constrained formulation of~\cite{mballo} or the classical hard-constrained reach-avoid set (which, in our setting, corresponds to an emergency landing flight envelope with a binary notion of safety), thereby providing a unified treatment of existing formulations.

\item We prove a general monotonicity result: if one violation cost function is more conservative than another, in the sense that it assigns no smaller cost to any state at any time, then the emergency landing flight envelope induced by the more conservative violation cost function is contained in that of the less conservative one. We also establish continuity of the flight envelope with respect to the cost function parameter. Finally, we show that suitable parameterized families in this class recover, in the limit as the parameter tends to infinity, either the soft-constrained formulation of~\cite{mballo} or the classical hard-constrained reach-avoid set (which, in our setting, corresponds to an emergency landing flight envelope with a binary notion of safety), thereby providing a unified treatment of existing formulations.

\item We develop a synthesis procedure that, for a given functional form of $\ell$, chooses the violation-cost parameter to meet a prescribed safety requirement on off-nominal operation while maximizing operational capability. Leveraging results from the sensitivity analysis, we reduce the design problem to a one-dimensional bracketing search and provide an algorithm with guaranteed convergence to the optimal parameter value.

\item We demonstrate the sensitivity results and synthesis algorithm in a fixed-wing emergency landing case study under propulsion failure in the presence of exogenous disturbances.

\end{itemize}

\subsection{Outline}

The remainder of the paper is organized as follows. Section~\ref{sec:problem} provides a formal definition of the emergency landing flight envelope. Building on this definition, Section~\ref{Soft_Constrained_VF} presents the soft-constrained HJ reachability framework used to characterize and compute the envelope. Section~\ref{sec:sensitivity} addresses question~\emph{(i)} via a sensitivity analysis over the one-parameter family of violation cost functions \(\{\ell_\lambda\}_{\lambda \ge 0}\). Question~\emph{(ii)} is addressed in Section~\ref{tuning}, which derives an algorithm for selecting \(\lambda\) to meet prescribed safety and performance requirements. Section~\ref{sec:landing_app} illustrates the practical implications of Sections~\ref{sec:sensitivity} and~\ref{tuning} using a fixed-wing emergency-landing scenario. Finally, Section~\ref{sec:conclusion} concludes the paper and discusses directions for future work.

\section{Problem Formulation}
\label{sec:problem}

In this section, we formulate emergency landing under propulsion failure as a robust reach-avoid problem with soft safety constraints.

\subsection{Robust Reachability as a Zero-Sum Game}\label{subsec:ra_zero_sum_game}

We model the aircraft under propulsion failure with time-invariant dynamics
\begin{equation}
  \dot{\mathrm{x}}(s)
  = f\bigl(\mathrm{x}(s),\bm a(s),\bm b(s)\bigr),
  \quad s\in[0,T],\quad \mathrm{x}(0)=x,
  \label{eq:ode}
\end{equation}

\noindent where we set \(s=0\) to be the time of the propulsion failure, \(x\in\mathbb{R}^n\) denotes the state at failure, and \(\mathrm{x}(\cdot):[0,T]\to\mathbb{R}^n\) is the resulting post-failure trajectory. The terminal time $T>0$ represents an upper bound on the allowable landing time: the aircraft must reach a safe landing site at some $\tau\in[0,T]$. The input function \(\bm a(\cdot)\) represents the post-failure control command, generated by the pilot/autopilot through the remaining control authority (e.g., remaining aerodynamic control surfaces), while \(\bm b(\cdot)\) models exogenous disturbances and unmodeled dynamics (e.g., wind gusts).

Let \(\mathbb{A}\subset\mathbb{R}^{m_a}\) and \(\mathbb{B}\subset\mathbb{R}^{m_b}\) be nonempty, compact input sets (that is, closed and bounded), and define the spaces of measurable signals 
\(
\mathcal{A}\coloneqq\{\bm a:[0,T]\!\to\!\mathbb{A}\ \big|\;\|\bm a\|_{L^\infty(0,T)} < \infty\},~
\mathcal{B}\coloneqq\{\bm b:[0,T]\!\to\!\mathbb{B}\ \big|\;\|\bm b\|_{L^\infty(0,T)} < \infty\}.
\) We assume that \(f:\mathbb{R}^n\times\mathbb{A}\times\mathbb{B}\to\mathbb{R}^n\) is bounded and uniformly continuous, and Lipschitz in the state vector uniformly in \((a,b) \in \mathbb{A} \times \mathbb{B}.\)

In this setting, Player~\(\mathrm{A}\) (the pilot/autopilot) selects \(\bm a\in\mathcal{A}\), while Player~\(\mathrm{B}\) (the disturbance) reacts via a nonanticipative strategy \(\delta\in\Delta\), yielding \(\bm b=\delta[\bm a]\), where
\begin{align}
\Delta &\coloneqq \{\delta:\mathcal{A}\!\to\!\mathcal{B} \mid \forall s\!\in\![0,T],\, \bm a,\bm{\bar a}\!\in\!\mathcal{A},\notag\\[-2pt]
&\quad \bm a\equiv\bm{\bar a}\ \text{a.e.\ on }\![0,s]\Rightarrow
\delta[\bm a]\equiv\delta[\bm{\bar a}]\ \text{a.e.\ on }[0,s]\}.
\label{eq:nonanticipative}
\end{align}

\noindent This nonanticipativity condition ensures that Player~\(\mathrm{B}\) cannot use Player~\(\mathrm{A}\)’s future inputs when choosing its current input. We study a zero-sum interaction in which Player~\(\mathrm{A}\) aims to guarantee a safe landing, whereas Player~\(\mathrm{B}\) acts adversarially.

Under the aforementioned assumptions, for any $\bm a\in\mathcal{A}$ and $\delta\in\Delta$, the induced pair $(\bm a,\delta[\bm a])$ generates a unique trajectory $\phi^{\bm a,\delta[\bm a]}_{x}(\cdot):[0,T]\to\mathbb{R}^n$ solving~\eqref{eq:ode} almost everywhere on $[0,T]$~\cite{c16}. In particular, if two control signals $\bm a,\bm{\bar a}\in\mathcal{A}$ agree almost everywhere on $[0,T]$, then under the nonanticipative strategy $\delta\in\Delta$ the corresponding disturbance signals $\delta[\bm a]$ and $\delta[\bm{\bar a}]$ also agree almost everywhere, and the resulting pairs $(\bm a,\delta[\bm a])$ and $(\bm{\bar a},\delta[\bm{\bar a}])$ induce the same trajectory. Thus, the nonanticipativity condition is naturally imposed up to almost-everywhere equivalence rather than pointwise equality, since differences on sets of measure zero do not affect the trajectory. Moreover, because admissible signals are only assumed to be measurable, the resulting trajectory need not be differentiable at every time; instead, it is understood in the Carathéodory sense, that is, as an absolutely continuous function satisfying the differential equation almost everywhere on $[0,T]$.

\subsection{Hard and Soft Safety Constraints}\label{subsec:Constraints}

To model the graded nature of safety, we consider nonempty, compact sets $\mathbb{C}_{2}\subseteq\mathbb{C}_{1}\subset\mathbb{R}^n$. The set $\mathbb{C}_{1}$ represents the admissible region of operation, enforced as a hard safety constraint. Its complement $\mathbb{C}_{1}^{\mathrm{c}}$ is the critical region in which the aircraft is not allowed to operate. The set $\mathbb{C}_{2}$ represents the nominal operating regime and is enforced as a soft safety constraint.

Accordingly, along the trajectory $\phi^{\bm a,\delta[\bm a]}_{x}(s)$ we
classify the operating condition as
\begin{equation}\label{eq:nom-deg-crit}
\begin{cases}
\phi^{\bm a,\delta[\bm a]}_{x}(s)\in\mathbb{C}_{2}
& \text{(nominal)}\\[2pt]
\phi^{\bm a,\delta[\bm a]}_{x}(s)\in\mathbb{C}_{1}\setminus\mathbb{C}_{2}
& \begin{array}{@{}l@{}}
\text{(degraded; permitted but}\\[-0.9ex]
\hspace{0.7 em}\text{incurs a penalty)}
\end{array}\\[2pt]
\phi^{\bm a,\delta[\bm a]}_{x}(s)\notin\mathbb{C}_{1}
& \text{(critical; forbidden).}
\end{cases}
\end{equation}
admissible region of operation

\noindent Trajectories are required to satisfy a hard state constraint by remaining in the admissible region of operation $\mathbb{C}_{1}$. Within $\mathbb{C}_{1}$, brief excursions from the nominal region $\mathbb{C}_{2}$ into the degraded region $\mathbb{C}_{1}\setminus \mathbb{C}_{2}$ are permitted but penalized. To quantify the cumulative cost associated with operating in $\mathbb{C}_{1}\setminus \mathbb{C}_{2}$, we introduce a nonnegative violation cost function $\ell : [0,T]\times\mathbb{R}^n \to [0,\infty)$ that specifies the instantaneous cost of operating in the degraded region as a function of time and state. We assume that \(\ell\) is bounded, continuous in \(s\), and Lipschitz continuous in \(x\), uniformly over \(s\in[0,T]\). Formally, for any $x \in \mathbb{R}^n$ and any $\tau \in [0,T]$, the cumulative violation cost incurred by the trajectory $\phi_{x}^{\bm a,\delta[\bm a]}(\cdot)$ up to time $\tau$ is given by
\begin{equation}\label{penalty}
    J(x,\tau;\bm a,\delta)
    \coloneqq \int_{0}^\tau \ell\bigl(s,\phi^{\bm a,\delta[\bm a]}_{x}(s)\bigr)\,ds.
\end{equation}

\noindent The cost \(J(x,\tau;\bm a,\delta)\) captures both how long the aircraft operates outside the nominal region \(\mathbb{C}_{2}\) and how severe that operation is. Longer operation in the degraded region and larger values of \(\ell(s,\phi_x^{\bm a,\delta[\bm a]}(s))\) both increase the cumulative cost; thus, states where \(\ell\) is large are interpreted as more damaging.

To formalize the soft safety constraint associated with the set \(\mathbb{C}_{2}\), we introduce a violation budget \(Q\) and impose \(J(x,\tau;\bm a,\delta)\leq Q\). Equivalently, a trajectory satisfies the soft safety constraint over \([0,\tau]\) if its cumulative violation cost does not exceed \(Q\). The parameter \(Q\in[0,\infty)\) specifies the maximum allowable cumulative cost associated with off-nominal operation for a given functional form of \(\ell\).

\subsection{Soft-Constrained Reach-Avoid Set}\label{subsec:Sets}

We now formalize our main object of interest, the emergency landing flight envelope, as the soft-constrained reach-avoid set.

Let \(\mathbb{T}\subset\mathbb{R}^n\) denote the landing target set, and assume that \(\mathbb{T}\) is nonempty and compact. In this study, \(\mathbb{T}\) is specified \emph{a priori} and consists of admissible touchdown states near ground level corresponding to a feasible landing site. In a more general setting, the landing site, and hence the associated target set \(\mathbb{T}\), could be identified as part of the proposed flight envelope computation method.

\begin{definition}\label{soft_constrained_reach_avoid_set}
Fix a budget $Q\in[0,\infty)$ and a cost function $\ell:[0,T]\times\mathbb{R}^n\to[0,\infty)$. The
\emph{soft-constrained reach-avoid set} is defined as
\begin{align}
{\mathcal{RA}}_{Q}&
\coloneqq \bigl\{x\in\mathbb{R}^{n}\bigm|
\forall\delta\in\Delta,\exists\,\bm{a}\in\mathcal{A},\nonumber\\
&\quad \exists\,\tau\in[0,T]\;:\;
\phi^{\bm{a},\delta[\bm{a}]}_{x}(\tau)\in\mathbb{T}\cap\mathbb{C}_{2},\nonumber\\
&\quad \hspace{5.7em} J(x,\tau;\bm a,\delta) \;\le\; Q, \label{eq:RA_soft_time_dep_family}\\ 
&\quad \forall\,s\in[0,\tau]\;:\;
\phi^{\bm{a},\delta[\bm{a}]}_{x}(s)\in \mathbb{C}_{1}\nonumber
\bigl\}.
\end{align}
\end{definition}

A state $x$ belongs to $\mathcal{RA}_{Q}$ if, starting from $x$, for every disturbance strategy $\delta\in\Delta$, Player~$\mathrm{A}$ can select a control policy $\bm a(\cdot)$ that guarantees the following: the trajectory $\phi^{\bm a,\delta[\bm a]}_{x}(\cdot)$ reaches a safe landing set within the nominal region, i.e., $\mathbb{T}\cap\mathbb{C}_{2}$, while satisfying both the hard and soft safety constraints. Equivalently, membership of $x$ in $\mathcal{RA}_{Q}$ is a necessary and sufficient condition for the existence of a robust control policy that guarantees a safe landing from $x$. Thus, $\mathcal{RA}_{Q}$ defines the emergency landing flight envelope, where the tolerance to off-nominal operation is governed by the violation cost function $\ell$ and the budget $Q$.

\section{Soft-Constrained Value Function}\label{Soft_Constrained_VF}

In this section, we use the framework of~\cite{mballo} to construct a value function whose zero sublevel set coincides with the emergency landing flight envelope introduced in Section~\ref{subsec:Sets}. This value function characterization provides a practical way to compute the flight envelope numerically by solving an HJ variational inequality (HJI-VI). This, in turn, provides a natural basis for the sensitivity analysis and synthesis procedures carried out in Sections~\ref{sec:sensitivity} and~\ref{tuning}, respectively.

Since the admissible operating region \(\mathbb{C}_{1}\), the nominal operating region \(\mathbb{C}_{2}\), and the landing region \(\mathbb{T}\) are compact subsets of \(\mathbb{R}^n\), each of these sets admits an implicit representation as the zero sublevel set of a signed distance function. In particular, we write
\begin{equation}
\mathbb{T} = \bigl\{ x \in \mathbb{R}^n \,\big|\, d_{\mathbb{T}}(x)\leq 0 \bigr\},
\label{eq:signed_distance}
\end{equation}
where $d_{\mathbb{T}}$ denotes the signed distance function associated with $\mathbb{T}$. Similarly,
\begin{equation}\label{equa:signed_distance_c1_c2}
\mathbb{C}_{1} = \bigl\{x \,\big|\, d_{\mathbb{C}_{1}}(x) \le 0\bigr\},\quad
\mathbb{C}_{2} = \bigl\{x \,\big|\, d_{\mathbb{C}_{2}}(x) \le 0\bigr\}.
\end{equation}

From Eqs.~\eqref{eq:signed_distance} and \eqref{equa:signed_distance_c1_c2}, membership of a state in \(\mathbb{C}_{1}\), \(\mathbb{C}_{2}\), or \(\mathbb{T}\) can be assessed directly from the sign of the associated signed distance function. Specifically, for each of these sets, a state lies inside it if the signed distance is negative, on the boundary if it is zero, and outside it if it is positive.

Using the signed distance function representations of \(\mathbb{C}_{1}\), \(\mathbb{C}_{2}\), and \(\mathbb{T}\) in Eqs.~\eqref{eq:signed_distance} and \eqref{equa:signed_distance_c1_c2}, we introduce a value function \(V\) that provides an implicit representation of~\eqref{eq:RA_soft_time_dep_family}. In the definition of \(V\) below, for any \(t\in[0,T]\), \(\mathcal{A}_t\) denotes the set of admissible control signals on \([t,T]\), \(\Delta_t\) denotes the associated class of nonanticipative strategies on \([t,T]\) (cf.~\eqref{eq:nonanticipative}), and \(\phi_{t,x}^{\bm a,\delta[\bm a]}(\cdot):[t,T]\to\mathbb{R}^n\) denotes the unique trajectory of~\eqref{eq:ode} with initial condition \(\mathrm{x}(t)=x\), induced by the control-disturbance pair \((\bm a,\delta[\bm a])\).

The value function \(V:[0,T]\times\mathbb{R}^{n}\times[0,\infty)\to\mathbb{R}\) is defined by
\begin{equation}\label{value_function_adjusted_family}
\begin{aligned}
V(t,x,Q)\coloneqq {}&\sup_{\delta \in \Delta_{t}}\inf_{\bm{a}\in\mathcal{A}_{t}}\min_{\tau \in [t, T]}
\max\Bigl\{ d_{\mathbb{C}_{2}}\bigl(\phi^{\bm{a},\delta[\bm{a}]}_{t,x}(\tau)\bigr),\\
&\int_{t}^{\tau} \ell\bigl(s,\phi^{\bm{a},\delta[\bm{a}]}_{t,x}(s)\bigr)\,ds - Q,\\ 
&\max_{s\in [t,\tau]} d_{\mathbb{C}_{1}}\bigl(\phi^{\bm{a},\delta[\bm{a}]}_{t,x}(s)\bigr),\;
d_{\mathbb{T}}\bigl(\phi^{\bm{a},\delta[\bm{a}]}_{t,x}(\tau)\bigr)
\Bigr\}.
\end{aligned}
\end{equation}
For a landing task initiated from state \(x\) at time \(t\) with budget \(Q\), the value function \(V\) evaluates, under the worst-case disturbance strategy \(\delta\in\Delta_t\) and the best available control input signal \(\bm a\in\mathcal A_t\), whether there exists a time \(\tau\in[t,T]\) at which the aircraft can reach the landing target within the nominal regime while satisfying both the hard and soft safety requirements. The supremum over \(\delta\) captures the worst-case effect of the disturbance, while the infimum over \(\bm a\) reflects the controller's best effort to steer the system toward a successful landing despite that disturbance. 

In particular, \(V\) simultaneously accounts for four quantities:

\begin{enumerate}
  \item[(i)] \label{term_1} the cumulative cost of off-nominal operation relative to the available budget through
  \begin{equation}\label{term1}
\begin{aligned}
\int_t^\tau \ell\bigl(s,\phi^{\bm a,\delta[\bm a]}_{t,x}(s)\bigr)\,ds - Q,
\end{aligned}
\end{equation}

  \item[(ii)] \label{term_2} violation of the admissible operating region \(\mathbb{C}_1\) through
\begin{equation}\label{term2}
\begin{aligned}
\max_{s\in[t,\tau]} d_{\mathbb{C}_1}\bigl(\phi^{\bm a,\delta[\bm a]}_{t,x}(s)\bigr),
\end{aligned}
\end{equation}

  \item[(iii)] \label{term_3} whether the aircraft trajectory lies in the nominal operating region \(\mathbb{C}_2\) at the candidate landing time \(\tau\) through
  
\begin{equation}\label{term3}
\begin{aligned}
d_{\mathbb{C}_2}\bigl(\phi^{\bm a,\delta[\bm a]}_{t,x}(\tau)\bigr),
\end{aligned}
\end{equation}

\item[(iv)] \label{term_4} whether the trajectory reaches the landing target set \(\mathbb{T}\) at that same time \(\tau\) through

\begin{equation}\label{term4}
\begin{aligned}
d_{\mathbb{T}}\bigl(\phi^{\bm a,\delta[\bm a]}_{t,x}(\tau)\bigr).
\end{aligned}
\end{equation}

\end{enumerate}

% (i) the cumulative cost of off-nominal operation relative to the available budget through
% \begin{equation}\label{term1}
% \begin{aligned}
% \int_t^\tau \ell\bigl(s,\phi^{\bm a,\delta[\bm a]}_{t,x}(s)\bigr)\,ds - Q,
% \end{aligned}
% \end{equation}
% (ii) violation of the admissible operating region \(\mathbb{C}_1\) through
% \begin{equation}\label{term2}
% \begin{aligned}
% \max_{s\in[t,\tau]} d_{\mathbb{C}_1}\bigl(\phi^{\bm a,\delta[\bm a]}_{t,x}(s)\bigr),
% \end{aligned}
% \end{equation}
% (iii) whether the aircraft trajectory lies in the nominal operating region \(\mathbb{C}_2\) at the candidate landing time \(\tau\) through
% \begin{equation}\label{term3}
% \begin{aligned}
% d_{\mathbb{C}_2}\bigl(\phi^{\bm a,\delta[\bm a]}_{t,x}(\tau)\bigr),
% \end{aligned}
% \end{equation}
% and (iv) whether the trajectory reaches the landing target set \(\mathbb{T}\) at that same time \(\tau\) through
% \begin{equation}\label{term4}
% \begin{aligned}
% d_{\mathbb{T}}\bigl(\phi^{\bm a,\delta[\bm a]}_{t,x}(\tau)\bigr).
% \end{aligned}
% \end{equation}

\noindent If there exists a time \(\tau\in[t,T]\) such that all four quantities defined in \eqref{term1}, \eqref{term2}, \eqref{term3}, and \eqref{term4} are nonpositive, then a safe landing is possible. To determine whether such a time exists, \(V\) takes the minimum over \(\tau\in[t,T]\) of the maximum of these four quantities.

The following proposition shows that \(\mathcal{RA}_{Q}\) can be characterized as the zero sublevel set of \(V\).

\begin{proposition}\label{prop_1}
Let \(V:[0,T]\times\mathbb{R}^{n}\times[0,\infty)\to\mathbb{R}\) be the value function defined in~\eqref{value_function_adjusted_family}. Then, for each \(Q\in[0,\infty)\), the soft-constrained reach-avoid set \(\mathcal{RA}_Q\) defined in~\eqref{eq:RA_soft_time_dep_family} satisfies
\begin{equation}\label{eq:RAQn-zero-sub}
\mathcal{RA}_{Q}
=\bigl\{x\in\mathbb{R}^{n}\ \big|\ V(0,x,Q)\le 0\bigr\}.
\end{equation}
\end{proposition}
\begin{proof}
See Appendix~\ref{apx:prop1}.
\end{proof}
It follows from~\eqref{eq:RAQn-zero-sub} that \(V(0,x,Q)\) is less than or equal to zero if and only if the state \(x\) belongs to \(\mathcal{RA}_Q\); otherwise \(V(0,x,Q)>0\). Equivalently, saying that $x$ is an element of $\mathcal{RA}_Q$ means that Player~$\mathrm{A}$ can guarantee reaching a safe landing region in $\mathbb{T}$ within $[0,T]$ while satisfying the hard and soft constraints against any admissible disturbance strategy of Player~$\mathrm{B}$.

To compute \(V\), we follow the procedure in~\cite{mballo}. The first step is to augment the system dynamics~\eqref{eq:ode} with an auxiliary state variable. Consider the augmented dynamics \(\hat{f}:[0, T]\times\mathbb{R}^n\times\mathbb{R}\times\mathbb{A}\times\mathbb{B}\to\mathbb{R}^n\times\mathbb{R}\),
\begin{equation}
\begin{bmatrix}\dot{\mathrm{x}}(s)\\ \dot{\mathrm{z}}(s)\end{bmatrix}
\!\!=\!\!\hat{f}\big(s,\mathrm{x}(s),\mathrm{z}(s),\!\bm a(s),\!\bm b(s)\big)
\!=\!\begin{bmatrix} f\big(\mathrm{x}(s),\bm a(s),\bm b(s)\big)\\ -\ell\big(s,\mathrm{x}(s)\big)\!\end{bmatrix}\!.
\label{eq:augmented_ode}
\end{equation}
\noindent where \(s\in[t,T]\) and the initial condition is \(\begin{bmatrix}\mathrm{x}(t)\\[2pt] \mathrm{z}(t)\end{bmatrix}
= \begin{bmatrix}x\\[2pt] z\end{bmatrix}\).

Here, the auxiliary state trajectory \(\mathrm{z}(\cdot):[t,T]\to\mathbb{R}\) tracks the remaining violation budget along the aircraft trajectory, where \(z\in[0,\infty)\) denotes the initial budget. Because \(\dot{\mathrm{z}}(s)=-\ell(s,\mathrm{x}(s))\), the remaining budget decreases whenever the aircraft operates outside the nominal regime.

Under the standing assumptions on \(f\) in Section~\ref{subsec:ra_zero_sum_game} and on \(\ell\) in Section~\ref{subsec:Sets}, the augmented dynamics \(\hat{f}\) are bounded and uniformly continuous, and Lipschitz in \((x,z)\) uniformly over \((s,a,b)\). Therefore, for any \(\bm a\in\mathcal{A}_{t}\) and \(\delta\in\Delta_{t}\), the pair \((\bm a,\delta[\bm a])\) induces a unique Lipschitz continuous augmented trajectory, denoted by \(\phi^{\bm a,\delta[\bm a]}_{t,x,z}(\cdot):[t,T]\to\mathbb{R}^n\times\mathbb{R}\), which solves~\eqref{eq:augmented_ode} almost everywhere on \([t,T]\)~\cite{c16}. This solution can be written as 
\({\phi}^{\bm{a},\delta[\bm a]}_{t,x,z} 
= \big(\phi^{\bm{a},\delta[\bm a]}_{t,x},\ \xi^{\bm{a},\delta[\bm a]}_{t,x,z}\big)\),
since the state variable \(x\) evolves independently of \(z\) in the augmented dynamics. The function \(\xi^{\bm{a},\delta[\bm a]}_{t,x,z}(\cdot)\) denotes the unique solution of the scalar ODE for \(z\). Using~\eqref{eq:augmented_ode}, the value function \(V\) can be expressed as
\begin{multline}
\!V(t,x,z)\!=\! 
\sup_{\delta \in \Delta_{t}}\;\!\!\inf_{\bm{a} \in \mathcal{A}_{t}}\;\!\!\min_{\tau \in [t, T]}\;\!\!
\max\Bigl\{\!
d_{\mathbb{C}_{2}}\bigl(\phi^{\bm{a},\delta[\bm{a}]}_{t,x}(\tau)\bigr),\\
\hspace{0.20em}{-\;\!\!\xi^{\bm{a},\delta[\bm{a}]}_{t,x,z}(\tau),\;\!\!
\max_{s \in [t,\tau]}\;\!\!\!d_{\mathbb{C}_{1}}\bigl(\phi^{\bm{a},\delta[\bm{a}]}_{t,x}(s)\bigr),\;\!
d_{\mathbb{T}}\bigl(\phi^{\bm{a},\delta[\bm{a}]}_{t,x}(\tau)\bigr)
\!\Bigr\},}
\label{eq:soft_value_function_rewritten}
\end{multline}

In~\eqref{eq:soft_value_function_rewritten}, the budget \(Q\) is encoded in the initial value of the augmented state; that is, we set \(z=Q\), or equivalently, \(\xi^{\bm a,\delta[\bm a]}_{t,x,z}(t)=z=Q\). Applying Bellman’s Principle of Optimality to \eqref{eq:soft_value_function_rewritten} yields that \(V\) is the unique viscosity solution of the HJI–VI stated in Theorem~\ref{final_theorem}.
\begin{theorem}\label{final_theorem}
The value function \(V\) defined in~\eqref{eq:soft_value_function_rewritten} is Lipschitz continuous and is the unique viscosity solution of the HJI variational inequality
\begin{align}\label{visco_solu}
0 &= \max\Bigl\{\,\!d_{\mathbb{C}_{1}}(x) - W(t,x,z),\;
\min\Bigl[ \max\{d_{\mathbb{C}_{2}}(x), -z, \nonumber\\
& \hspace*{0.3em} d_{\mathbb{T}}(x)\}\!-\!V(t,x,z),\;
\frac{\partial V}{\partial t}
+ \!H\bigl(t, x, z, \frac{\partial V}{\partial{x}}, \frac{\partial V}{\partial{z}}\bigr) \Bigr] \Bigr\},
\end{align}
for \((t,x,z)\in[0,T)\times\mathbb{R}^n\times[0,\infty)\), with terminal condition
\begin{equation}\label{viscosity_solu2}
V(T,x,z)=\max\{d_{\mathbb{C}_{1}}(x),\,d_{\mathbb{C}_{2}}(x),\,d_{\mathbb{T}}(x),\,-z\}.
\end{equation}
The Hamiltonian is
\begin{equation}\label{eq:hamiltonian}
H(t,x,z,p_x,p_z)
= \min_{a\in\mathbb{A}}\max_{b\in\mathbb{B}}
\Bigl\{ p_x\cdot f(x,a,b) - p_z\,\ell(t,x) \Bigr\}.
\end{equation}
\end{theorem}

The existence and uniqueness of the viscosity solution to the HJI-VI in Theorem~\ref{final_theorem} follow by arguments analogous to those in Theorem~3 of~\cite{mballo} and standard viscosity-solution results for differential games~\cite{Evans}; for brevity, we omit the proof. The HelperOC toolbox~\cite{Toolbox_2} and the Level Set Toolbox~\cite{Mitchell2008ToolboxLS} can be used to numerically compute the value function \(V\) by solving the HJI-VI in Theorem~\ref{final_theorem}.

% The Lipschitz continuity of \(W\) follows from the regularity of the augmented dynamics~\eqref{eq:augmented_ode} and from the Lipschitz continuity (in \(x\)) of the signed distance functions defining \(\mathbb{T}\), \(\mathbb{C}_1\), and \(\mathbb{C}_2\); see, e.g.,~\cite{Evans}. Furthermore, existence and uniqueness of the viscosity solution to the HJI-VI in Theorem~\ref{final_theorem} follow by arguments analogous to those in Theorem~3 of~\cite{mballo}; for brevity, we omit the proof.

\section{Sensitivity of the Emergency Landing Flight Envelope to the Violation Cost Function}\label{sec:sensitivity}

% The choice of the violation cost function $\ell$ is aircraft-dependent; for instance, one may parameterize $\ell$ as a weighted sum of basis functions, \( \ell_{\boldsymbol{\lambda}}(s,x)=\sum_{k=1}^{K}\lambda_k\,\ell_k(s,x) \)
% (e.g., polynomial bases), or via more general nonlinear compositions, and tune the coefficient vector
% \(\boldsymbol{\lambda}=(\lambda_1,\ldots,\lambda_K)^\top \in [0, \infty)^K\) (manually or algorithmically) to set a desired tolerance to off-nominal operation (e.g., as a function of remaining fatigue life). 

The choice of the violation cost function $\ell$ is aircraft-dependent, since tolerance to off-nominal operation may vary with vehicle type or condition. To capture this variability, we consider a parameterized family of violation cost functions $\{\ell_\lambda\}_{\lambda \in [0,\infty)}$ drawn from a broader admissible class $\mathcal{L}$, where the parameter $\lambda$ encodes the level of tolerance to degraded operation. For clarity of exposition, we focus on the scalar-parameter case.\footnote{The scalar parameterization induces a total order on $[0,\infty)$. Analogous sensitivity conclusions extend to multi-parameter families $\ell_{\boldsymbol{\lambda}}$, $\boldsymbol{\lambda}\in[0,\infty)^K$, along any componentwise monotone path in $[0,\infty)^K$.}

% \footnote{We pose this study around a single parameter to obtain a total ordering over the parameter set. Several of our results consider limits as $\lambda$ approaches the boundary of its domain. One may view these results as existing along a line segment through a larger parameter space; this interpretation suggests immediate generalizations.}.

% we consider a one-parameter family of cost functions \(\{\ell_\lambda\}_{\lambda \in [0, \infty)}\),\footnote{We pose this study around a single parameter to obtain a more natural ordering and monotonicity over the parameter set. Several of our results consider limits as $\lambda$ approaches the boundary of its domain. One may view these results as existing along a line segment through a larger parameter space; this interpretation suggests immediate generalizations, which we leave to future work to preserve conceptual clarity.} drawn from a broader admissible class \(\mathcal{L}\).

We define the class \(\mathcal{L}\) of admissible families of violation cost functions as follows.

\begin{definition}\label{def:class_of_families}
Let \(\Lambda\subseteq[0,\infty)\). The class \(\mathcal{L}\) consists of all families
\(\{\ell_\lambda\}_{\lambda\in\Lambda}\) with \(\ell_\lambda:[0,T]\times\mathbb{R}^n\to[0,\infty)\) such that:
\begin{enumerate}
\item[(i)] For all \(\lambda\in\Lambda\) and \(t\in[0,T]\),
\[
\ell_\lambda(t,x)=0 \ \text{for } x\in \mathbb{C}_{2},
\qquad
\ell_\lambda(t,x)>0 \ \text{for } x\in \mathbb{C}_{2}^c.
\]
\item[(ii)] For all \(\lambda_1<\lambda_2\), \(t\in[0,T]\), and \( x \in \mathbb{C}_1 \setminus \mathbb{C}_2\),
\[
\ell_{\lambda_1}(t,x)\le \ell_{\lambda_2}(t,x).
\]
\item[(iii)] Each \(\ell_\lambda\) is continuous on \([0,T]\times\mathbb{R}^n\), and there exist constants
\(M,L<\infty\) such that, for all \(\lambda\in\Lambda\), \(t\in[0,T]\), and \(x,y\in\mathbb{R}^n\),
\[
0\le \ell_\lambda(t,x)\le M,
\qquad
|\ell_\lambda(t,x)-\ell_\lambda(t,y)| \le L \|x-y\|.
\]
Moreover, \(\lambda\mapsto \ell_\lambda(t,x)\) is uniformly continuous on \(\Lambda\), with a modulus of continuity independent of \((t,x)\).
\end{enumerate}
\end{definition}

By Definition~\ref{def:class_of_families}, for any $\lambda\in\Lambda$, the function $\ell_\lambda$ assigns zero violation cost to states in $\mathbb{C}_{2}$ and a strictly positive cost to states in its complement $\mathbb{C}_{2}^{\,c}$. Importantly, aside from the monotonicity requirement in $\lambda$, we impose no further structure on the map $\lambda \mapsto \ell_\lambda$. 

As an illustrative example, one may construct $\ell_\lambda$ so that states located farther from the boundary $\partial\mathbb{C}_{2}$ incur higher violation cost than states closer to the boundary. In this case, the distance from $\partial\mathbb{C}_{2}$ determines the spatial variation of the cost, while $\lambda$ serves as a tuning parameter that controls the conservativeness of $\ell_\lambda$: larger $\lambda$ corresponds to a more conservative (more severe) violation cost function, whereas smaller $\lambda$ corresponds to a less conservative (less severe) one.

Let $V_\lambda$ denote the value function obtained from~\eqref{value_function_adjusted_family} by replacing $\ell$ with $\ell_\lambda$. In particular, $V_\lambda$ can be computed by solving the HJ PDE in Theorem~\ref{final_theorem} with the augmented dynamics $\hat f_\lambda$, defined by substituting $\ell_\lambda$ for $\ell$ in~\eqref{eq:augmented_ode}. Similar to the procedure in Section~\ref{Soft_Constrained_VF}, the corresponding soft-constrained reach-avoid set, denoted by $\mathcal{RA}_{Q;\lambda}$, can be characterized in terms of $V_\lambda$ as follows.

\begin{proposition}\label{prop:char-Wn}
Fix $\lambda\in[0,\infty)$ and $Q\in[0,\infty)$. Then
\begin{equation}\label{eq:RAQn-zero-sublevel}
\mathcal{RA}_{Q;\lambda}
=\bigl\{x\in\mathbb{R}^{n}\ \big|\ V_{\lambda}(0,x,Q)\le 0\bigr\}.
\end{equation}
\end{proposition}
\noindent The proof of Proposition~\ref{prop:char-Wn} follows the same structure as that of Proposition~\ref{prop_1}.

We study the sensitivity of the emergency landing flight envelope to the violation cost function by varying $\lambda$ to induce a family $\{\ell_\lambda\}_{\lambda\in\Lambda} \in \mathcal{L}$, ranging from less to more conservative cost functions.

\begin{theorem}\label{thm:RA-mono}
Let $\{\ell_\lambda\}_{\lambda\in\Lambda} \in \mathcal{L}$ and let $Q\in[0,\infty)$.
Fix $\lambda \in \Lambda$ and let $(\lambda_k)_{k\ge1} \subset \Lambda$ with $\lambda_k \to \lambda$.
Then,
\begin{enumerate}
  \item[(i)] \label{itm:mono-up}
  If $\lambda_k\uparrow \lambda$, then
  \begin{equation}\label{eq:RA-mono-up}
  \mathcal{RA}_{Q;\lambda_1}\supseteq \mathcal{RA}_{Q;\lambda_2}\supseteq\cdots\supseteq \mathcal{RA}_{Q;\lambda}.
  \end{equation}

  \item[(ii)] \label{itm:mono-down}
  If $\lambda_k\downarrow \lambda$, then
  \begin{equation}\label{eq:RA-mono-down}
  \mathcal{RA}_{Q;\lambda_1}\subseteq \mathcal{RA}_{Q;\lambda_2}\subseteq\cdots\subseteq \mathcal{RA}_{Q;\lambda}.
  \end{equation}

  \item[(iii)] \label{itm:hausdorff}
  The sequence converges in the Hausdorff metric~\footnote{For nonempty subsets $A,B \subset \mathbb{R}^n$, the Hausdorff metric is defined as
\[
d_H(A,B)
\coloneqq
\max\left\{
\sup_{a\in A}\inf_{b\in B}\|a-b\|,
\;
\sup_{b\in B}\inf_{a\in A}\|b-a\|
\right\},
\]
where $\|\cdot\|$ denotes the Euclidean norm.    }:
  \begin{equation}\label{eq:RA-Hausdorff}
  d_H\!\left(\mathcal{RA}_{Q;\lambda_k},\,\mathcal{RA}_{Q;\lambda}\right)\xrightarrow[k\to\infty]{}0.
  \end{equation}
\end{enumerate}
\end{theorem}

\begin{proof}
See Appendix~\ref{apx:Thm2}.
\end{proof}

Theorem~\ref{thm:RA-mono} establishes that the soft-constrained reach-avoid set \(\mathcal{RA}_{Q;\lambda}\) is monotone with respect to the violation cost parameter \(\lambda\): as \(\lambda\) increases, \(\mathcal{RA}_{Q;\lambda}\) shrinks, whereas as \(\lambda\) decreases, \(\mathcal{RA}_{Q;\lambda}\) can only expand.
This result is intuitive. Increasing \(\lambda\) assigns a larger violation cost to states in the degraded region \(\mathbb{C}_1 \setminus \mathbb{C}_2\), which in turn increases the cumulative violation cost along trajectories that traverse this region. As a result, fewer states admit a robust policy that guarantees reaching the landing set while satisfying the hard and soft safety constraints, yielding a smaller \(\mathcal{RA}_{Q;\lambda}\). Conversely, when \(\lambda\) is reduced, operations in \(\mathbb{C}_1 \setminus \mathbb{C}_2\) are penalized less severely, so additional states may become feasible, thereby allowing
\(\mathcal{RA}_{Q;\lambda}\) to expand outward.

In addition, Theorem~\ref{thm:RA-mono} shows that the mapping \(\lambda \mapsto \mathcal{RA}_{Q;\lambda}\) is continuous in the Hausdorff metric. This property precludes abrupt and substantial changes in the geometry of \(\mathcal{RA}_{Q;\lambda}\) under small variations of \(\lambda\). In particular, as \(\lambda\) varies, the boundary of \(\mathcal{RA}_{Q;\lambda}\) evolves continuously in the Hausdorff sense, without jumps or discontinuities.

We next examine the asymptotic behavior of the soft-constrained reach–avoid set induced by the family $\{\ell_\lambda\}_{\lambda\in\Lambda} \in \mathcal{L}$. In particular, our objective is to characterize the limiting set and to investigate the notion of set convergence under which this limit arises. For this analysis, we take $\Lambda = [0,\infty)$ and consider the limit as $\lambda$ goes to infinity.

%%%%%%%%%%%%%%%%%%%%%%%%%%%%%%%%%%%%%%%%%%%%%%%%%%%%%%%%%%%%%%%%%%%%%%%%%%%%%%%%%%%%%%%%%%%%%%%%%%%%%%%%%%%%%%%%%%%%%%%%%%%%%%%%%%%%%%%%%%%%%%%%%%%%%%%%%%%%%%%%%%%%%%%%%%%%%%%%%%%%%%%%%%
%%%%%%%%%%%%%%%%%%%%%%%%%%%%%%%%%%%%%%%%%%%%%%%%%%%%%%%%%%%%%%%%%%%%%%%%%%%%%%%%%%%%%%%%%%%%%%%%%%%%%%%%%%%%%%%%%%%%%%%%%%%%%%%%%%%%%%%%%%%%%%%%%%%%%%%%%%%%%%%%%%%%%%%%%%%%%%%%%%%%%%%%%%

Let \(\{\ell_\lambda\}_{\lambda\in[0,\infty)} \subset \mathcal{L} \)
be a parameterized family of violation cost functions.
Let $\ell_{\infty}$ denote the pointwise limit of $\ell_\lambda$ as $\lambda \to \infty$, i.e.,
\(
\ell_{\infty}(t,x) \coloneqq \lim_{\lambda\to\infty} \ell_\lambda(t,x).
\)
Let $\mathcal{RA}_{Q;\lambda}$ and $\mathcal{RA}_{Q;{\infty}}$ denote the soft-constrained reach-avoid sets corresponding to the cost functions $\ell_\lambda$ and $\ell_{\infty}$, respectively.
\begin{theorem}\label{prop:penalty_function1}
Let $\lambda \in [0,\infty)$ and $Q \in [0,\infty)$. Then
\begin{equation}
\mathcal{RA}_{Q;\lambda} \;\supseteq\; \mathcal{RA}_{Q;{\infty}},
\label{eq:hausdorff_conv}
\end{equation}
and furthermore,
\begin{equation}
\mu\!\left(\mathcal{RA}_{Q;\lambda} \triangle \mathcal{RA}_{Q;{\infty}}\right)
\;\xrightarrow[\lambda\to\infty]{}\; 0,
\label{eq:measure_conv}
\end{equation}
where $\triangle$\footnote{For sets $A,B \subset \mathbb{R}^n$, $A \triangle B \coloneqq (A \setminus B) \cup (B \setminus A)$.}
denotes symmetric difference and $\mu$ is the Lebesgue measure on $\mathbb{R}^n$.
\end{theorem}
\begin{proof}
See Appendix~\ref{prop:penalty_function1a}.
\end{proof}

% Let \(\{\ell_\lambda\}_{\lambda\in[0,\infty)}\) be any family in \(\mathcal{L}\).
% Let \(h\) denote the pointwise limit of \(\ell_\lambda\) as \(\lambda \to \infty\), i.e.,
% \(
% h(t,x) \coloneqq \lim_{\lambda\to\infty}\ell_\lambda(t,x).
% \)
% We denote by \(\mathcal{RA}_{Q;h}\) the soft-constrained reach-avoid set obtained by using
% \(h\) as the violation cost function.

% \begin{theorem}\label{prop:penalty_function1}
% Let \(\lambda \in [0, \infty)\), \(x\in\mathbb{R}^{n}\), and \(Q\in[0,\infty)\). Then
% \begin{equation}
% \mathcal{RA}_{Q;\,\lambda} \;\supseteq\; \mathcal{RA}_{Q;\,h},
% \label{eq:hausdorff_conv}
% \end{equation}
% and furthermore,
% \begin{equation}
% \mu\!\left(\,\mathcal{RA}_{Q;\,\lambda}\,\triangle\,\mathcal{RA}_{Q;\,h}\,\right)
% \;\xrightarrow[\lambda\to\infty]{}\;0,
% \label{eq:measure_conv}
% \end{equation}
% where $\mu(.)$ denotes the Lebesgue measure on $\mathbb{R}^{n}$.~\footnote{For sets $A,B \subset \mathbb{R}^n$, $A \triangle B \coloneqq (A \setminus B) \cup (B \setminus A)$.}
% % \footnote{The symbol
% % \(\triangle\) denotes the symmetric difference of two sets,
% % i.e., \(\mathcal{RA}_{Q;\,\lambda} \triangle \mathcal{RA}_{Q;\,h} = (\mathcal{RA}_{Q;\,\lambda} \setminus \mathcal{RA}_{Q;\,h}) \cup (\mathcal{RA}_{Q;\,h} \setminus \mathcal{RA}_{Q;\,\lambda})\).
% % }
% \end{theorem}
Theorem~\ref{prop:penalty_function1} shows that pointwise convergence of $\ell_\lambda$ to $\ell_{\infty}$ is sufficient to ensure that, as $\lambda \to \infty$, the volume of the region in the state space where $\mathcal{RA}_{Q;\lambda}$ and $\mathcal{RA}_{Q;{\infty}}$ differ vanishes. In this sense, $\mathcal{RA}_{Q;\lambda}$ provides an over-approximation of $\mathcal{RA}_{Q;{\infty}}$. This approximation becomes especially valuable when the intended violation cost function (e.g. \(\ell_{\infty}\)) is discontinuous, as the corresponding value function inherits this discontinuity. The standard Hamilton-Jacobi theory~\cite{crandall1983} used to derive Theorem~\ref{final_theorem} cannot be applied to discontinuous value functions. From a computational standpoint, a discontinuous value function leads to numerical instability when computing the soft-constrained reach-avoid set.

The following remark illustrates this case by introducing a practically relevant discontinuous violation cost function \(\ell_{\infty}\) and a family \(\{\ell_{\lambda}\}_{\lambda\in[0,\infty)} \in \mathcal{L}\) that converges pointwise to \(\ell_{\infty}\), and therefore can be used to approximate \(\mathcal{RA}_{Q;{\infty}}\).

\begin{remark}\label{remark}
\medskip
Let \(\ell_{\infty}:[0,T]\times\mathbb{R}^n \to [0,\infty)\) be a violation cost function defined by
\begin{equation}\label{eq:ell_inf_structure}
\ell_{\infty}(s,x)
\coloneqq
\begin{cases}
0, & x\in \mathbb{C}_{2},\\[2mm]
M, & x\in \mathbb{C}_{2}^{\,c},
\end{cases}
\qquad M\in(0,\infty].
\end{equation}
\begin{enumerate}
\item If \(M\) is finite, \(\mathcal{RA}_{Q;\,\ell_{\infty}}\) generalizes the soft-constrained formulation of~\cite{mballo}. In particular, the choice \(M=1\) recovers the indicator-type penalty used in~\cite{mballo}, which penalizes time spent operating in the degraded region.

\item If \(M=\infty\), then any operation in the degraded region incurs infinite cost and is therefore inadmissible. Consequently, \(\mathcal{RA}_{Q;\,\ell_{\infty}}\) coincides with the classical hard-constrained reach-avoid set, which enforces a purely binary notion of safety.
\end{enumerate}
\end{remark}

Among the many possible constructions, we consider the parameterized family \(\{\ell_\lambda\}_{\lambda \ge 0}\) defined by
\begin{equation}\label{eq:penalty_exp_piecewise}
\ell_{\lambda}(s, x) =
\begin{cases}
0, & x \in \mathbb{C}_{2}, \\[6pt]
1 - \exp\!\Big(-K(\lambda)\dfrac{d_{\mathbb{C}_{2}}(x)}{\alpha}\Big),
& x \in \mathbb{C}_{2}^{\,c},
\end{cases}
\end{equation}
where \(\alpha\) denotes the Hausdorff distance between \(\mathbb{C}_{1}\) and \(\mathbb{C}_{2}\) with respect to the \(L_\infty\) norm, and \(K:[0,\infty)\to(0,\infty)\) is a positive, nondecreasing polynomial function of \(\lambda\). This family is used in the simulation experiments of Section~\ref{sec:landing_app}.

This family is motivated by applications in which the severity of off-nominal operation increases with the distance from the nominal region \(\mathbb{C}_2\), whereas sufficiently small excursions outside \(\mathbb{C}_2\) pose only a minor safety concern. For example, the nominal region of operation may be defined to include a safety margin from the aerodynamic stall boundary. In such a setting, slight departures from the nominal region may have little effect on stall risk, whereas larger excursions bring the aircraft closer to stall, where nonlinear aerodynamic effects become increasingly pronounced and adverse. The exponential form of the cost function captures this relationship by assigning near-zero cost to small deviations and causing the cost to increase rapidly toward its maximum as the distance from \(\mathbb{C}_2\) grows. The function \(K(\lambda)\) determines how rapidly this transition occurs. For a given application, \(\ell_\lambda\) may be designed so that its shape reflects how the severity of off-nominal operation varies across the state space, while \(K(\lambda)\) controls how conservative the resulting cost is.

One may readily verify that the family \(\{\ell_\lambda\}_{\lambda\ge 0}\) defined in \eqref{eq:penalty_exp_piecewise} belongs to \(\mathcal{L}\) and that \(\ell_\lambda\) converges pointwise to the discontinuous cost \(h\) as \(\lambda\to\infty\). More broadly, by allowing a general class of cost functions, the proposed framework accommodates a wide range of practically relevant violation cost profiles, thereby enabling application-specific choices without modifying the underlying framework.

\section{Synthesis via Parametric Tuning of the Emergency Landing Flight Envelope}\label{tuning}

This section translates the sensitivity results of Section~\ref{sec:sensitivity} into a synthesis procedure for selecting the violation cost parameter so as to satisfy an explicit safety requirement on off-nominal operation while maximizing operational capability. This synthesis procedure uses the performance metric $P$ and the safety metric $S$ to enforce a desired balance between operation in the nominal and degraded regimes.
\subsection{Performance and safety metrics}\label{subsec:PS_metrics}

Let $\mathcal{K}(\mathbb{R}^n)$ denote the collection of compact subsets of $\mathbb{R}^n$. We consider two functionals
$P,S:\mathcal{K}(\mathbb{R}^n)\to\mathbb{R}$ that provide, respectively, performance and safety measures for the soft-constrained reach-avoid set $\mathcal{RA}_{Q;\lambda}$. In particular, $P(\mathcal{RA}_{Q;\lambda})$ measures the operational capability retained by $\mathcal{RA}_{Q;\lambda}$ under the violation cost function $\ell_\lambda$. Typical examples of the operational capability quantified by such a metric include the maximum recoverable altitude, admissible airspeed range, or the volume of $\mathcal{RA}_{Q;\lambda}$. On the other hand, $S(\mathcal{RA}_{Q;\lambda})$ quantifies the degree to which a safe landing relies on off-nominal regimes, such as operation within a degraded region. For instance, $S(\mathcal{RA}_{Q;\lambda})$ may measure the fraction (or volume) of the reach-avoid set that lies within a degraded region of the state space. Alternatively, $S(\mathcal{RA}_{Q;\lambda})$ may capture the worst-case accumulated exposure to degraded operation along any guaranteed landing trajectory starting from $\mathcal{RA}_{Q;\lambda}$.

We impose only mild structural assumptions on these functionals, as stated in Assumption~\ref{assump:PS_properties}.

\begin{assumption}\label{assump:PS_properties}
The functionals $P,S:\mathcal{K}(\mathbb{R}^n)\to\mathbb{R}$ satisfy:
\begin{enumerate}
\item[(i)] \textbf{Monotonicity:} if $\mathcal{A}\subseteq\mathcal{B}$, then
$P(\mathcal{A}) \le P(\mathcal{B})$ and $S(\mathcal{A}) \le S(\mathcal{B})$.
\item[(ii)] \textbf{Continuity:} $P$ and $S$ are continuous with respect to the Hausdorff metric on $\mathcal{K}(\mathbb{R}^n)$.
\end{enumerate}
\end{assumption}

\subsection{Synthesis problem}\label{subsec:synthesis_problem}

Given a prescribed safety requirement $\bar S$, we seek to select $\lambda$ so as to maximize operational capability while ensuring that reliance on off-nominal operation does not exceed the specified tolerance:
\begin{equation}\label{eq:synthesis_opt}
\lambda^\star \in \arg\max_{\lambda \in \Lambda} P(\mathcal{RA}_{Q;\lambda})
\quad \text{s.t.} \quad
S(\mathcal{RA}_{Q;\lambda}) \le \bar S .
\end{equation}

We assume that the prescribed safety requirement $\bar S$ is feasible, i.e., there exists at least one $\lambda \in \Lambda$ such that $S(\mathcal{RA}_{Q;\lambda}) \le \bar S$. The optimization problem~\eqref{eq:synthesis_opt} admits a particularly simple characterization. The key observation is that both the reachable set $\mathcal{RA}_{Q;\lambda}$ and the safety metric $S(\mathcal{RA}_{Q;\lambda})$ vary monotonically with the violation cost parameter $\lambda$. As a result, the safety constraint defines a one-dimensional feasible set that is closed and connected, reducing the synthesis problem to a maximization over this interval.

By Theorem~\ref{prop:char-Wn}, the set-valued map $\lambda \mapsto \mathcal{RA}_{Q;\lambda}$ is continuous with respect to the Hausdorff metric. Under Assumption~\ref{assump:PS_properties}, the set functional $S(\cdot)$ is continuous on $\mathcal{K}(\mathbb{R}^n)$, and hence the induced scalar map $\lambda \mapsto S(\mathcal{RA}_{Q;\lambda})$ is continuous on $\Lambda$. It follows that the feasible set
\begin{equation}\label{eq:feasible_set}
\Lambda_{\mathrm{feas}}(\bar S)
\coloneqq \{\lambda \in \Lambda \mid S(\mathcal{RA}_{Q;\lambda}) \le \bar S\}
\end{equation}
is closed, as it is the preimage of the closed set $(-\infty,\bar S]$ under a continuous map. Moreover, since $\mathcal{RA}_{Q;\lambda}$ is monotone decreasing in $\lambda$ and $S(\cdot)$ is monotone increasing with respect to set inclusion, the map $\lambda \mapsto S(\mathcal{RA}_{Q;\lambda})$ is nonincreasing. Therefore, $\Lambda_{\mathrm{feas}}(\bar S)$ is a closed interval.

Consequently,~\eqref{eq:synthesis_opt} is equivalent to the following optimization over the feasible set $\Lambda_{\mathrm{feas}}(\bar S)$:

\begin{equation}\label{eq:synthesis_opt_feas}
\lambda^\star \in \arg\max_{\lambda \in \Lambda_{\mathrm{feas}}(\bar S)}
P(\mathcal{RA}_{Q;\lambda}).
\end{equation}
Since $P(\mathcal{RA}_{Q;\lambda})$ is nonincreasing in $\lambda$ and $\Lambda_{\mathrm{feas}}(\bar S)$ is a closed interval, an optimal solution is attained at the smallest feasible parameter (corresponding to the least conservative choice), i.e.,

\begin{equation}\label{eq:lambda_star}
\lambda^\star
=
\min \Lambda_{\mathrm{feas}}(\bar S).
\end{equation}

\subsection{Synthesis Algorithm}
\label{subsec:synthesis_problem}

Equation~\eqref{eq:lambda_star} reduces the synthesis problem~\eqref{eq:synthesis_opt} to identifying the smallest value of $\lambda$ that satisfies the prescribed safety requirement. Because the feasible set $\Lambda_{\mathrm{feas}}(\bar S)$ is a one-dimensional closed and connected set, and the map $\lambda \mapsto S(\mathcal{RA}_{Q;\lambda})$ is continuous and monotone, $\lambda^{*}$ can be computed reliably using a one-dimensional bracketing search (e.g., bisection) with guaranteed convergence.

To this end, define the scalar function
\begin{equation}\label{eq:g_lambda}
g(\lambda) \coloneqq S(\mathcal{RA}_{Q;\lambda}) - \bar S .
\end{equation}

\noindent We define the zero-level set of $g$ as
\begin{equation}\label{eq:zero_level_set}
\mathcal{Z}
\coloneqq
\{\lambda \in \Lambda \mid g(\lambda)=0\}.
\end{equation}
Since the map $\lambda \mapsto S(\mathcal{RA}_{Q;\lambda})$ is continuous and monotone nonincreasing, the set $\mathcal{Z}$ is a closed interval of $\mathbb{R}$. This interval corresponds to a plateau on which the safety constraint is exactly satisfied.

Let $\underline{\lambda} \coloneqq \min \mathcal{Z}$ and $\overline{\lambda} \coloneqq \max \mathcal{Z}$. Then, the monotonicity of $g$ implies the following sign structure:
\begin{equation}\label{eq:g_sign_structure}
\begin{aligned}
g(\lambda) &> 0 \quad &&\text{for all } \lambda < \underline{\lambda},\\
g(\lambda) &= 0 \quad &&\text{for all } \lambda \in
[\underline{\lambda},\overline{\lambda}],\\
g(\lambda) &< 0 \quad &&\text{for all } \lambda > \overline{\lambda}.
\end{aligned}
\end{equation}

Therefore, the feasible set can be expressed as
\begin{equation}\label{eq:feasible_set_g}
\Lambda_{\mathrm{feas}}(\bar S)
=
\{\lambda \in \Lambda \mid g(\lambda)\le 0\}
=
[\underline{\lambda},\infty)\cap\Lambda.
\end{equation}
In particular, the smallest feasible parameter coincides with the left endpoint
of the zero-level set, i.e.,
\begin{equation}\label{eq:lambda_star_from_g}
\lambda^\star
=
\min \Lambda_{\mathrm{feas}}(\bar S)
=
\underline{\lambda}.
\end{equation}

In practice, $\underline{\lambda}$ can be computed in two steps using standard one-dimensional bracketing methods (e.g., the bisection method or the Chandrupatla method~\cite{Chandrupatla1997}). First, apply the bisection method to an interval $[\lambda_{\min},\lambda_{\max}]\subset\Lambda$ satisfying $g(\lambda_{\min})>0$ and $g(\lambda_{\max})<0$ to obtain a parameter $\tilde{\lambda}\in\mathcal{Z}$, i.e., $g(\tilde{\lambda})=0$; the continuity of $g$ and the sign pattern in~\eqref{eq:g_sign_structure} guarantee that bisection converges to such a parameter. Second, since the zero-level set $\mathcal{Z}$ is (potentially) an interval, a second bisection on $[\lambda_{\min},\tilde{\lambda}]$ is used to locate its left endpoint $\underline{\lambda}=\min\mathcal{Z}$, which equals $\lambda^\star$. This procedure is summarized in Algorithm~\ref{alg:synthesis}. 

\begin{algorithm}[H]
\caption{Synthesis of the Violation-Cost Parameter $\lambda^\star$}\label{alg:synthesis}
\begin{algorithmic}
\STATE \textbf{Input:} Safety tolerance $\bar S$, search interval $[\lambda_{\min},\lambda_{\max}]\subset\Lambda$
\STATE \textbf{Output:} Synthesized parameter $\lambda^\star$
\STATE 
\STATE Define $g(\lambda) \coloneqq S(\mathcal{RA}_{Q;\lambda}) - \bar S$
\STATE Define $\mathcal{Z} \coloneqq \{\lambda \in \Lambda \mid g(\lambda)=0\}$
\STATE
\STATE \textbf{Step 1: Find a point in the zero-level set}
\STATE \hspace{0.5cm}Apply bisection on $[\lambda_{\min},\lambda_{\max}]$ to obtain $\tilde{\lambda}\in\mathcal{Z}$
\STATE
\STATE \textbf{Step 2: Locate the left endpoint}
\STATE \hspace{0.5cm}Apply bisection on $[\lambda_{\min},\tilde{\lambda}]$ to locate $\underline{\lambda}\coloneqq\min\mathcal{Z}$
\STATE \hspace{0.5cm}Set $\lambda^\star \gets \underline{\lambda}$
\STATE
\STATE \textbf{return} $\lambda^\star$
\end{algorithmic}
\end{algorithm}

\section{Application: Safe Landing of a Fixed-wing Aircraft}
\label{sec:landing_app}

In this section, we compute the emergency landing flight envelope for a fixed-wing aircraft under propulsion failure using the cost function~\eqref{eq:penalty_exp_piecewise}. Specifically, we evaluate \(\mathcal{RA}_{Q;\lambda}\) for different values of \(\lambda\) to illustrate the conclusions of Theorem~\ref{thm:RA-mono} and Theorem~\ref{prop:penalty_function1}, and thus demonstrate the sensitivity analysis developed in Section~\ref{sec:sensitivity}. This numerical study also illustrates the synthesis procedure in Algorithm~\ref{alg:synthesis}.

\subsection{Problem Setup}
We consider the fixed-wing point-mass model with propulsion failure introduced in~\cite{mballo}.
\begin{equation}
\begin{bmatrix}
    \dot{V_a} \\
    \dot{\gamma} \\
    \dot{h}
\end{bmatrix}
=
\frac{1}{M}
\begin{bmatrix}
    -D(\alpha, V_a) - M g_0 \sin(\gamma) +  F_{\text{dist}} \\
    \frac{L(\alpha, V_a)}{V} - \frac{M g_0}{V_a} \cos(\gamma) \\
    MV\sin(\gamma) 
\end{bmatrix}.
\label{Aircraft_dynamics}
\end{equation}

\noindent Here, \(V_a\) is the airspeed \((\mathrm{m/s})\), \(\gamma\) the flight path angle \((\mathrm{deg})\), and \(h\) the altitude \((\mathrm{m})\); \(M\) denotes the aircraft mass \((\mathrm{kg})\), \(g_0\) the gravitational acceleration \((\mathrm{m/s}^2)\), and \(\alpha\) the angle of attack \((\mathrm{deg})\). The lift and drag forces are denoted by \(L\) and \(D\), respectively, both measured in newtons \((\mathrm{N})\), and \(F_{\mathrm{dist}}\) \((\mathrm{N})\) represents exogenous disturbances such as wind. All parameter values for this aircraft model are given in~\cite{bayen}.

The admissible and nominal operating regions are denoted by \(\mathbb{C}_1\) and \(\mathbb{C}_2\), respectively. These sets satisfy \(\mathbb{C}_2 \subset \mathbb{C}_1\), as defined in~\eqref{eq:constraint} and~\eqref{eq:constraint1}. We compute the emergency landing flight envelope over \(h \in [0,36]\), corresponding to a low-altitude propulsion-loss scenario. Higher-altitude cases can be analyzed by extending the range of \(h\).

% We restrict the altitude to $h\in[0,36]$ to reflect the assumption that the propulsion-loss event occurs at most $36\,\mathrm{m}$ above ground level.
\begin{equation}\label{eq:constraint}
\mathbb{C}_1 \!=\!\{(V_a,\gamma,h)\mid 61 \le V_a \le 84,\; -3 \le \gamma \le 0,\; 0 \le h \le 36\}.
\end{equation}
\begin{equation}\label{eq:constraint1}
\mathbb{C}_2\!=\!\{(V_a,\gamma,h)\mid 66 \le V_a \le 79,\; -3 \le \gamma \le 0,\; 0 \le h \le 36\}.
\end{equation}

The boundary of the admissible operating region, \(\mathbb{C}_1\), is determined by aerodynamic and structural considerations, such as stall limits and structural integrity constraints. Operation outside \(\mathbb{C}_1\) is unsafe, whereas operation within the nominal region maintains a margin from the boundary of \(\mathbb{C}_1\), as suggested by the Federal Aviation Administration~\cite{FAR1990}. This margin offers multiple benefits: it reduces stall risk and helps preserve the aircraft's structural health by limiting prolonged operation near the boundary of \(\mathbb{C}_1\). For this application, an exponential distance-based cost function such as~\eqref{eq:penalty_exp_piecewise} is well suited, since slight excursions outside the nominal region may still pose only a minor safety concern, whereas larger deviations correspond to increasingly severe off-nominal operation. 

In the absence of propulsive forces, the objective is to characterize the set of states from which the aircraft can achieve a safe touchdown with a vertical descent rate of at most $0.91~\mathrm{m/s}$, using the angle of attack as the sole control input. This touchdown condition is formalized by the target set
\begin{equation}\label{Target}
\begin{aligned}
\mathbb{T} \coloneqq \{(V_a,\gamma,h) \mid {}&
66 \le V_a \le 79,\; -0.62 \le \gamma \le 0,\\
& 0 \le h \le 0.5 \}.
\end{aligned}
\end{equation}
\noindent The set \(\mathbb{T}\) represents admissible touchdown states associated with a feasible landing zone on the ground.

To encode the safety specification for this problem, we enforce remaining within the admissible region of operation $\mathbb{C}_1$ as a hard constraint, while operation in the degraded region $\mathbb{C}_1 \setminus \mathbb{C}_2$ is permitted but incurs a cost specified by the violation cost function~\eqref{eq:penalty_exp_piecewise}. The soft safety constraint resulting from the cost function~\eqref{eq:penalty_exp_piecewise} captures the aircraft's ability to temporarily depart from nominal operating conditions in order to retain sufficient maneuverability to achieve a safe landing under this emergency flight condition.

\subsection{Computed Emergency Landing Flight Envelope}

The left panel of Fig.~\ref{RA_sets} shows the three-dimensional $(V_a,\gamma,h)$ state space considered in this study. The nominal region $\mathbb{C}_2$ is enclosed by a blue boundary, and the degraded region $\mathbb{C}_1\setminus\mathbb{C}_2$ is shown in light red with a boundary comprising both blue and red segments. The target set $\mathbb{T}$ (landing zone) is shown in gray with a black boundary.

\begin{figure*}[!t]
\centering
\begin{minipage}[t]{0.329\textwidth}
  \centering
  \includegraphics[width=\linewidth]{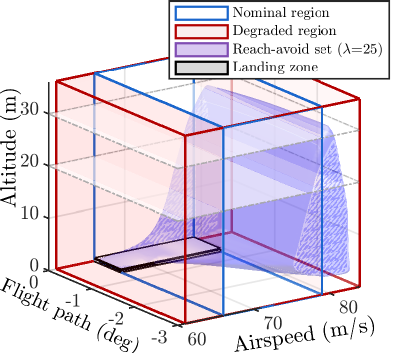}
\end{minipage}\hfill
\begin{minipage}[t]{0.329\textwidth}
  \centering
  \includegraphics[width=\linewidth]{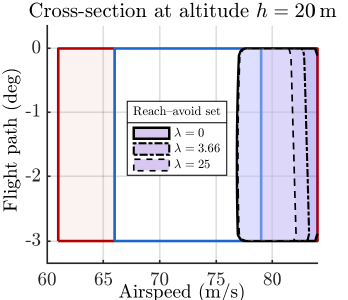}
\end{minipage}\hfill
\begin{minipage}[t]{0.329\textwidth}
  \centering
  \includegraphics[width=\linewidth]{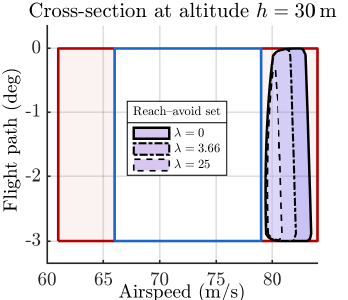}
\end{minipage}
\caption{Computed emergency landing flight envelope $\mathcal{RA}_{Q;\lambda}$ for $\lambda \in \{0,\,3.66,\,25\}$ with $Q=5$.
\emph{Left:} Three-dimensional $(V_a,\gamma,h)$ state space showing the nominal region $\mathbb{C}_2$ (closed region enclosed by the blue boundary), the degraded region $\mathbb{C}_1\setminus\mathbb{C}_2$ (light red region with a boundary comprising both blue and red segments), and the landing target set $\mathbb{T}$ (grey region with black boundary).
\emph{Middle and right:} Two-dimensional slices of $\mathcal{RA}_{Q;\lambda}$ (purple region with black boundary) at $h=20~\mathrm{m}$ and $h=30~\mathrm{m}$ for $\lambda \in \{0,\,3.66,\,25\}$, showing that the emergency landing flight envelope shrinks monotonically as $\lambda$ increases (Theorem~\ref{thm:RA-mono}).}
\label{RA_sets}
\end{figure*}

We use the HelperOC toolbox~\cite{Toolbox_2} to solve the HJI variational inequality in Theorem~\ref{final_theorem} for the augmented dynamics obtained by combining the aircraft model~\eqref{Aircraft_dynamics} with the violation cost function~\eqref{eq:penalty_exp_piecewise}, thereby computing the value function \(V_{\lambda}\). The computation is performed backward in time until numerical convergence is reached; in our simulations, this occurs after approximately \(10~\mathrm{s}\) of backward propagation. We then recover \(\mathcal{RA}_{Q;\lambda}\) via Proposition~\ref{prop:char-Wn}.

The left panel of Fig.~\ref{RA_sets} shows the soft-constrained reach-avoid set $\mathcal{RA}_{Q;\lambda}$ for $\lambda=25$ and $Q=5$ (purple region). States inside $\mathcal{RA}_{5;25}$ correspond to initial conditions from which the aircraft can reach the landing zone while remaining in the admissible region of operation $\mathbb{C}_1$, i.e., either in the nominal region $\mathbb{C}_2$ or in the degraded region $\mathbb{C}_1\setminus\mathbb{C}_2$. Moreover, landings from these initial conditions are guaranteed to incur an accumulated cost for operation in the degraded region that does not exceed the budget $Q=5$.

The middle and right panels of Fig.~\ref{RA_sets} further illustrate the structure of the soft-constrained reach-avoid set as $\lambda$ varies by showing two-dimensional slices at altitudes $h=20$ and $h=30$~m. For each altitude value, the corresponding slice of $\mathcal{RA}_{Q;\lambda}$ is shown as a purple region with a black boundary. The reduction in size of these slices as $\lambda$ increases visually corroborates the monotonicity property established in Theorem~\ref{thm:RA-mono}: larger values of $\lambda$ assign higher cost to departures from the nominal region, resulting in a strictly smaller reach-avoid set. In particular, as $\lambda$ increases from $0$ to $25$, the emergency landing flight envelope shrinks across all altitude slices, reflecting reduced tolerance for degraded operation. It is important to note that, at the higher-altitude slice ($h=30~\mathrm{m}$), the states from which a safe landing is achievable lie entirely within the degraded region for all $\lambda \in \{0,\,3.66,\,25\}$, indicating that safe landing from higher altitude is only possible by operating outside the nominal envelope. This observation highlights the importance of adopting a graded safety formulation that distinguishes nominal and degraded regimes and enables an appropriate balance of operation across them, rather than a purely binary safe-unsafe characterization.

The asymptotic behavior of $\mathcal{RA}_{Q;\lambda}$ as $\lambda\to\infty$ is governed by the pointwise convergence of the cost $\ell_\lambda$~\eqref{eq:penalty_exp_piecewise} on $\mathbb{C}_1$, as established by Theorem~\ref{prop:penalty_function1} and Remark~\ref{remark}. In the limit $\lambda\to\infty$, $\ell_\lambda$ converges pointwise on $\mathbb{C}_1$ to the indicator function $\mathbf{1}_{\mathbb{C}_1\setminus\mathbb{C}_2}$. Consequently, $\mathcal{RA}_{Q;\lambda}$ approaches, as $\lambda$ increases, the soft-constrained reach-avoid set introduced in~\cite{mballo}, with $\mathbf{1}_{\mathbb{C}_1\setminus\mathbb{C}_2}$ serving as the violation cost function.

To illustrate this asymptotic behavior, Fig.~\ref{convergence} plots the highest safe altitude of $\mathcal{RA}_{Q;\lambda}$ as a function of $\lambda$ for budgets $Q\in\{0,0.2,2.0,5.0,8.0,10.0\}$. For each $Q$, the resulting curve is monotone nonincreasing in $\lambda$ and converges to a finite limit consistent with the highest safe altitude of the soft-constrained reach-avoid set in~\cite{mballo}.

 \begin{figure}[!t]
\centering
\includegraphics[width=0.8\columnwidth]{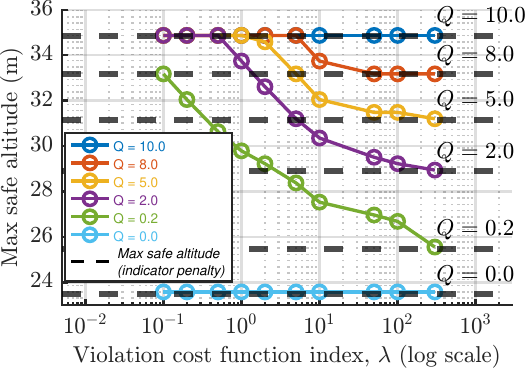}
\caption{Asymptotic behavior of the soft-constrained reach-avoid set $\mathcal{RA}_{Q;\lambda}$ as $\lambda$ increases. The plot reports the highest safe altitude of $\mathcal{RA}_{Q;\lambda}$ as a function of $\lambda$ for $Q\in\{0,0.2,2.0,5.0,8.0,10.0\}$. For each $Q$, the resulting curve is monotone nonincreasing in $\lambda$ and levels off as $\lambda\to\infty$, converging to the highest safe altitude of the reach-avoid set defined in~\cite{mballo}, which uses the indicator function $\mathbf{1}_{\mathbb{C}_1\setminus\mathbb{C}_2}$ as the violation cost function.}
\label{convergence}
\end{figure}

\subsection{Parameter Selection via One-Dimensional Root Finding}

We apply Algorithm~\ref{alg:synthesis} to compute the parameter $\lambda$ in~\eqref{eq:penalty_exp_piecewise} to achieve a desired balance between operational capability and safety. To this end, we consider the performance and safety metrics $P,S:\;2^{\mathbb{R}^n}\to\mathbb{R}_{\ge 0}$ defined as
\begin{align}
P\!\big(\mathcal{RA}_{Q;\lambda}\big)
&\coloneqq \max_{(V,\gamma,h)\in \mathcal{RA}_{Q;\lambda}} \; h,\\[4pt]
S\!\big(\mathcal{RA}_{Q;\lambda}\big)
&\coloneqq \frac{\mu\!\big(\mathcal{RA}_{Q;\lambda}\cap(\mathbb{C}_1\setminus\mathbb{C}_2)\big)}
{\mu\!\big(\mathbb{C}_1\setminus\mathbb{C}_2\big)}.
\end{align}

The metric $P(\mathcal{RA}_{Q;\lambda})$ quantifies operational capability by measuring the maximum altitude from which a safe landing can be guaranteed. Larger values of $P$ correspond to a larger emergency landing flight envelope and a higher maximum safe altitude, and thus to greater operational capability. The metric $S(\mathcal{RA}_{Q;\lambda})$ quantifies safety by measuring the fraction of the degraded region $\mathbb{C}_1\setminus\mathbb{C}_2$ that is contained in the emergency landing flight envelope. Larger values of $S$ indicate greater tolerance for operation in degraded conditions, whereas smaller values correspond to increased conservativeness with respect to degraded operation. In particular, $S=0$ indicates that no states in $\mathbb{C}_1\setminus\mathbb{C}_2$ are included in $\mathcal{RA}_{Q;\lambda}$, whereas $S=1$ indicates that $\mathbb{C}_1\setminus\mathbb{C}_2$ is entirely contained in $\mathcal{RA}_{Q;\lambda}$.

We enforce the safety requirement $S(\mathcal{RA}_{Q;\lambda}) \le \bar{S}$ with $\bar{S}=0.35$, i.e., at most $35\%$ of the degraded region $\mathbb{C}_1\setminus\mathbb{C}_2$ is contained in the emergency landing flight envelope. We then solve the optimization problem~\eqref{eq:synthesis_opt} to maximize $P(\mathcal{RA}_{Q;\lambda})$ subject to this constraint.

% Figure~\ref{Param_Tuning} 
 Fig.~\ref{Param_Tuning} illustrates the successive iterates of Algorithm~\ref{alg:synthesis} in terms of the candidate parameter values $\lambda$ and the corresponding metrics $P(\mathcal{RA}_{Q;\lambda})$ and $S(\mathcal{RA}_{Q;\lambda})$. As $\lambda$ increases, both metrics are monotone nonincreasing, consistent with Assumption~\ref{assump:PS_properties}. The algorithm converges to the parameter value $\lambda=3.66$, which satisfies the prescribed safety requirement while maximizing performance. The corresponding performance and safety values are $P=33.46~\mathrm{m}$ and $S=0.35$, respectively.

From a control synthesis perspective, the resulting value \(\lambda = 3.66\) induces a set of robust, safe controllers that steer the aircraft to the landing zone set while guaranteeing that any operation in the degraded region is confined to a restricted subset whose volume is at most \(35\%\) of the total volume of that region.

\begin{figure}[!t]
    \centering
    \includegraphics[width=0.8\columnwidth]{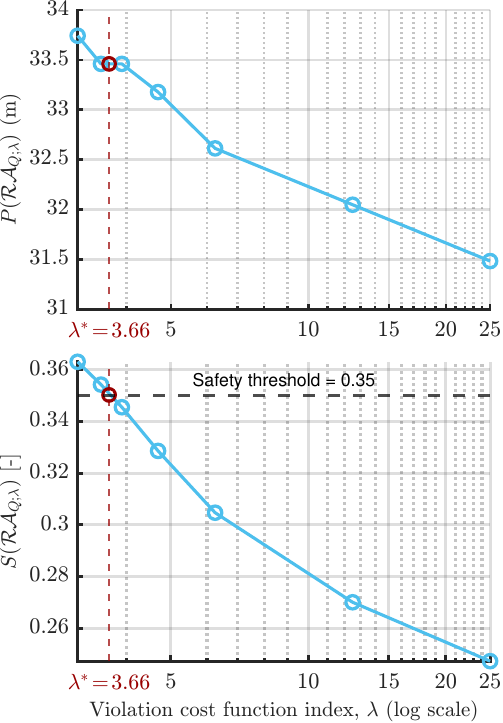}
    \caption{Parameter tuning via Algorithm~\ref{alg:synthesis}. The figure depicts the iterative update of the candidate parameter $\lambda$ and the corresponding values of the performance and safety metrics. Both metrics vary monotonically with $\lambda$, and the algorithm converges to the parameter value $\lambda=3.66$. At this value, the safety constraint is active (i.e., $S(\mathcal{RA}_{Q;\lambda})=\bar{S}=0.35$), yielding the performance-maximizing parameter subject to the prescribed safety requirement.}
    \label{Param_Tuning}
\end{figure}

\section{Conclusion}\label{sec:conclusion}
This paper develops a Hamilton–Jacobi reachability framework for computing emergency landing flight envelopes under a graded safety model, rather than a binary safe/unsafe formulation. The graded safety specification is captured by soft and hard safety constraints. The soft constraint is defined by a continuous, designer-specified violation cost function that encodes the severity of off-nominal operation and a finite violation budget $Q$, which bounds the allowable extent of off-nominal operation. The hard constraint prohibits operation in the critical regime.

Our main theoretical results establish that, for a broad class of state- and time-dependent violation cost functions, the emergency landing envelope varies monotonically with the cost assigned to off-nominal operation. Furthermore, for any parameterized family in this class, the envelope varies continuously with the cost parameter in the Hausdorff sense. These sensitivity results provide a principled link between safety conservatism and operational capability. They also preclude abrupt geometric changes in the envelope under small parameter variations. Building on these properties, we formulate a synthesis problem to select the parameter value that meets a prescribed safety requirement and solve it using a one-dimensional bracketing algorithm with guaranteed convergence.

Numerical experiments on a fixed-wing point-mass aircraft under propulsion failure with exogenous disturbances illustrate how the envelope shrinks as the cost parameter increases and characterize its asymptotic behavior as the parameter tends to infinity. The experiments also demonstrate algorithmic tuning of the cost parameter to enforce a prescribed constraint on degraded-regime operation, quantified by the envelope volume over the degraded region, while maximizing recoverable altitude.

Several directions merit further investigation. First, developing a systematic methodology for designing the violation cost function \(\ell\) for a given application is an important next step, since the choice of \(\ell\) determines how off-nominal operation is quantified and should ideally reflect application-specific notions of risk, performance degradation, or accumulated damage. Second, extending the synthesis algorithm to multi-parameter families (e.g., $\boldsymbol{\lambda}\in\mathbb{R}^K$) would enable richer structure in the violation cost function. Third, incorporating time-varying or state-dependent budgets $Q$ could better capture evolving system health and phase-dependent operating limits throughout the emergency landing maneuver. Fourth, studying alternative safety metrics based on worst-case trajectories (rather than set volume) would enable parameter tuning that is defined explicitly at the trajectory level. Finally, applying the framework to higher fidelity aircraft models remains an important direction. Since grid-based solutions of HJ equations scale poorly with state dimension, future work should investigate scalable approximation techniques, such as neural network-based value function representations, decomposition methods, or sampling-based approximations, to enable application to high dimensional aircraft models. Combining these scalable computational approaches with broader uncertainty descriptions, (e.g., parametric aerodynamic uncertainty and wind fields), would help evaluate computational feasibility in even more operationally realistic settings.

% \section*{Acknowledgments}
% This should be a simple paragraph before the References to thank those individuals and institutions who have supported your work on this article.

\appendices

\section{Proof of Proposition 1}\label{apx:prop1}
\begin{proof}
Pick \(Q\in [0, \infty)\). We show that: (i) \(x\!\in\!\mathcal{RA}_{Q;}\!\Rightarrow\! V(0,x,Q)\le 0\); 
(ii) \(V(0,x,Q)\!\le\!0\!\Rightarrow\! x \in\mathcal{RA}_{Q}\).\\
\textit{(i)} Suppose, for contradiction, that \(x \!\in\!\mathcal{RA}_{Q}\) and \(V(0,x,Q)>0\). Define the function \(C\) as follows 
\begin{equation}
\begin{aligned}[t]
& C(0,x,Q,\tau,\delta,\bm{a}) \!\!\coloneqq \!\!\max\Big\{ \!d_{\mathbb{C}_{2}}\!\big(\phi^{\bm{a},\delta[\bm{a}]}_{x}(\tau)\big),\; d_{\mathbb{T}}\!\big(\phi^{\bm{a},\delta[\bm{a}]}_{x}(\tau)\big),\\
&\hspace{1em} \int_0^\tau \ell\!\big(s,\phi^{\bm{a},\delta[\bm{a}]}_{x}(s)\big)\,ds - Q,\; \max_{s\in[0,\tau]} d_{\mathbb{C}_{1}}\!\big(\phi^{\bm{a},\delta[\bm{a}]}_{x}(s)\big)\Big\}.
\end{aligned}
\end{equation}
We then have
\begin{equation}
V(0,x,Q)
=\sup_{\delta \in \Delta}\ \inf_{\bm a \in \mathcal{A}}\ \min_{\tau \in [0,T]}
C(0,x,Q,\tau,\delta,\bm a) \,>\, 0.
\end{equation}
Therefore, there exist \(\epsilon>0\) and \(\bar{\delta}\in\Delta\) such that
\begin{equation}
\inf_{\bm a \in \mathcal{A}}\ \min_{\tau \in [0,T]}
C(0,x,Q,\tau,\bar{\delta},\bm a) \,>\, \epsilon,
\end{equation}
and hence, for all \(\bm a\in\mathcal{A}\) and all \(\tau\in[0,T]\),
\begin{equation}
C(0,x,Q,\tau,\bar{\delta},\bm a)>\epsilon.
\end{equation}

Given that \(x \in\mathcal{RA}_{Q;\,\lambda}\), for every \(\delta\in\Delta\) there exist \(\bm{a}_\delta\in\mathcal{A}\) and \(\tau_\delta\in[0,T]\) such that the trajectory \(\phi^{\bm{a}_\delta,\delta[\bm{a}_\delta]}_{x}(\cdot)\) satisfies the reach–avoid conditions in~\eqref{eq:RA_soft_time_dep_family}, and consequently \(C(0,x,Q,\tau_{\bar{\delta}},\bar{\delta},\bm{a}_{\bar{\delta}}) \le 0\). This contradicts the earlier conclusion that \(C(0,x,Q,\tau,\bar{\delta},\bm{a}) > \epsilon > 0\) for all \(\bm{a} \in \mathcal{A}\) and \(\tau \in [0,T]\). Therefore,  \(x \!\in\!\mathcal{RA}_{Q}\!\Rightarrow\! V(0,x,Q)\le 0\) \\

\textit{(ii)} Pick \(x\) such that \(V(0,x,Q) \le 0\) and suppose, for contradiction, that \(x \notin \mathcal{RA}_{Q}\). Then, by the assumption \(x \notin \mathcal{RA}_{Q}\), there exists \(\eta > 0\) and \(\tilde{\delta} \in \Delta\) such that \(\forall\, \bm{a} \in \mathcal{A},\; \forall\, \tau \in [0,T],\; C(0,x,Q,\tau,\tilde{\delta},\bm{a}) > \eta\).

Since \(V(0,x,Q) \le 0\), for all \(\epsilon > 0\) there exist \(\tilde{\bm{a}} \in \mathcal{A}\) and \(\tilde{\tau} \in [0,T]\) such that \(C(0,x,Q,\tilde{\tau},\delta,\tilde{\bm{a}}) \le \epsilon\). Letting \( \epsilon = \eta/2 \) and \( \delta = \tilde{\delta} \) leads to a contradiction. Therefore, \( V(0,x,Q) \leq 0 \;\Rightarrow\;  x \in\mathcal{RA}_{Q}\). 

\textit{(ii)} Pick \(x\) such that \(V(0,x,Q)\le 0\) and suppose, for contradiction, that \(x\notin\mathcal{RA}_Q\).
By the definition of \(\mathcal{RA}_Q\), there exist \(\eta>0\) and \(\tilde{\delta}\in\Delta\) such that
\begin{equation}\label{Ctd_Eq}
\inf_{\bm a\in\mathcal{A}}\ \min_{\tau\in[0,T]}
C(0,x,Q,\tau,\tilde{\delta},\bm a)\ \ge\ \eta .
\end{equation}
Since \(V(0,x,Q) \le 0\), for all \(\epsilon > 0\) there exist \(\tilde{\bm{a}} \in \mathcal{A}\) and \(\tilde{\tau} \in [0,T]\) such that \(C(0,x,Q,\tilde{\tau},\delta,\tilde{\bm{a}}) \le \epsilon\). Choosing \(\epsilon=\eta/2\) and \(\delta=\tilde{\delta}\) yields a contradiction with \eqref{Ctd_Eq}. Therefore, \(V(0,x,Q)\le 0\) implies \(x\in\mathcal{RA}_Q\).

From \textit{(i)} and \textit{(ii)}, we conclude that
\begin{equation}\label{eq:RA_def}
\mathcal{RA}_{Q}
= \bigl\{ x \in \mathbb{R}^n \;\big|\; V(0,x,Q) \le 0 \bigr\}.
\end{equation}
\end{proof}
\section{Proof of Theorem 2}\label{apx:Thm2}
\begin{proof}
\noindent\textbf{(i)}  It suffices to show that \( \mathcal{RA}_{Q;\lambda_1}\!\!\!\supseteq\!\!\mathcal{RA}_{Q;\lambda_2} \).
Fix \( 0\!\!\le\!\!\lambda_1\!\!<\!\!\lambda_2\!\!<\!\!\infty \), \( Q\in[0,\infty) \), and \( x\in\mathbb{R}^n \).

Let \(C_{\lambda_2}\) be defined analogously to \(C\) in the proof of Proposition 1, with \(\ell\) replaced by \(\ell_{\lambda_2}\).
If \(x \in \mathcal{RA}_{Q;\lambda_2}\), then there exist a control signal \(\bm a \in \mathcal{A}\) and a time \(\tau \in [0,T]\) such that
\begin{equation}
C_{\lambda_2}(0,x,Q,\tau,\delta,\bm a) \le 0 \quad \text{for all } \delta \in \Delta.
\end{equation}

\noindent In particular, for any \(s \in [0,\tau]\), we have
\(\phi^{\bm a,\delta[\bm a]}_{x}(s) \in \mathbb{C}_1\).
Since \(\ell_{\lambda_1}(s,y) \le \ell_{\lambda_2}(s,y)\) for all \(s \in [0,T]\) and \(y \in \mathbb{C}_1\), it follows that
\begin{equation}
\int_0^\tau \ell_{\lambda_{1}}\!\big(s,\phi^{\bm a,\delta[\bm a]}_{t,x}(s)\big)\,ds
\le
\int_0^\tau \ell_{\lambda_{2}}\!\big(s,\phi^{\bm a,\delta[\bm a]}_{t,x}(s)\big)\,ds
\le Q.
\end{equation}

\noindent Hence,
\begin{align}
V_{\lambda_1}(0,x,Q)
&=\sup_{\delta\in\Delta}\inf_{\bm a\in\mathcal{A}}\min_{\tau\in[0,T]}
C_{\lambda_1}(0,x,Q,\tau,\delta,\bm a) \nonumber\\
&\le
\sup_{\delta\in\Delta}\inf_{\bm a\in\mathcal{A}}\min_{\tau\in[0,T]}
C_{\lambda_2}(0,x,Q,\tau,\delta,\bm a)
\label{eq:lambda_monotonicity}\\
&= V_{\lambda_2}(0,x,Q).
\nonumber
\end{align}

\noindent Since \(Q\) is arbitrary, it follows that
\(\mathcal{RA}_{Q;\lambda_1} \supseteq \mathcal{RA}_{Q;\lambda_2}\) for all \(Q\ge 0\).\\
\noindent\textbf{(ii)} The proof follows the same steps as in \textbf{(i)} and is omitted for brevity.\\
\noindent\textbf{(iii)} Since $V_\lambda$ is Lipschitz continuous in $(t,x,Q)$, there exists $L_t>0$ such that $|V_\lambda(t,x,Q)-V_\lambda(s,x,Q)|\le L_t|t-s|$ for all $t,s\in[0,T]$ and all $x\in\mathbb{R}^n$. Hence $\{x:V_\lambda(0,x,Q)\le 0\}\subseteq
S_T\coloneqq\{x:V_\lambda(T,x,Q)\le L_tT\}$.
Because $S_T$ is contained in a compact ball $\mathbb{K}$, the zero sublevel set of $V_\lambda(0,\cdot,Q)$ is contained in the compact set $\mathbb{K}$.

Define the set value function \(
F:\Lambda \to 2^{K}\) 
\begin{equation}\label{set_value_func}
F(\psi)\;\coloneqq\;\mathcal{RA}_{Q;\psi}
\;.
\end{equation}
\noindent Let \(\psi \in \Lambda\) and consider any sequence \((\psi_k)_{k \ge 1} \subset \Lambda\) such that \(\psi_k \to \psi\).
For each \(k\), pick \(x_k \in F(\psi_k)\).
By compactness of \(\mathbb{K}\), there exists a subsequence \((x_{k_j})_{j \ge 1}\) such that \(x_{k_j} \to x\) for some \(x \in \mathbb{K}\). Since $\psi$ only affects the cost $\ell_\psi$ (and the system dynamics~\eqref{eq:ode}
are independent of $\psi$), for any $\psi,\tilde\psi\in\Lambda$ we have the bound
\begin{equation}\label{eq:W-psi-perturb}
\begin{aligned}
|V_\psi(t,x,Q)-V_{\tilde\psi}(t,x,Q)|
&\le T \sup_{(s,y)\in[0,T]\times\mathbb{R}^n}
|\Delta_{\psi,\tilde\psi}(s,y)|.
\end{aligned}
\end{equation}
where $\Delta_{\psi,\tilde\psi}(s,y) \coloneqq \ell_\psi(s,y)-\ell_{\tilde\psi}(s,y)$. Therefore, the uniform continuity of the map $\psi\mapsto \ell_\psi$ (with a modulus independent of $(t,x)$; see Definition~\ref{def:class_of_families})
implies that $\psi\mapsto V_\psi(t,x,Q)$ is uniformly continuous on $\Lambda$, uniformly over $(t,x,Q)\in[0,T]\times\mathbb{R}^n\times[0,\infty)$. Hence, we can write
\begin{equation}
\begin{aligned}
&\bigl|V_{\psi_{k_j}}(0,x_{k_j},Q)-V_{\psi}(0,x,Q)\bigr| \\
&\le \underbrace{\bigl|V_{\psi_{k_j}}(0,x_{k_j},Q)-V_{\psi_{k_j}}(0,x,Q)\bigr|}_{\le L\,(\|x_{k_j}-x\|)} \\[-0.25em]
&\quad+\underbrace{\bigl|V_{\psi_{k_j}}(0,x,Q)-V_{\psi}(0,x,Q)\bigr|}_{\le \omega(|\psi_{k_j}-\psi|)}
\;\xrightarrow[j\to\infty]{}\;0.
\end{aligned}
\end{equation}
where \(L\) is the Lipschitz constant of \(V_{\psi}\) in \(x\) and
\(\omega\) is the modulus of continuity for the map \(\psi\mapsto V_\psi(t,x,Q)\). Therefore
\begin{equation}
\lim_{j\to\infty} V_{\psi_{k_j}}(0,x_{k_j},Q) \;=\; V_{\psi}(0,x,Q).
\end{equation}

\noindent Since \!\(V_{\psi_{k_j}}(0,x_{k_j},Q)\le \! 0\) for all \(j\), passing to the limit yields \(V_{\psi}(0,x,Q)\le 0\). Hence $x\in F(\psi)$. Consequently, every limit point of sequences $x_{k}\in F(\psi_k)$ with $\psi_k\to\psi$ belongs to $F(\psi)$. Therefore $F$ is upper hemicontinuous at $\psi$. Since $\psi\in\Lambda$ was arbitrary, we conclude that $F$ is upper hemicontinuous on $\Lambda$.

Now we prove lower hemicontinuity of $F$ on $\Lambda$.
Let $\phi \in \Lambda$ and let $\{\psi_j\}_{j\geq 1} \subset \Lambda$
be any sequence such that $\psi_j \to \phi$. Pick any $x \in F(\phi)$. 

For any $j \ge 1$, we have
\begin{equation}\label{eq:W-modulus-lambda}
\bigl|V_{\psi_j}(0,x,Q) - V_{\phi}(0,x,Q)\bigr|
\le \omega\!\bigl(|\psi_j - \phi|\bigr).
\end{equation}

\noindent Since $V_{\phi}(0,\cdot,Q)$ is continuous in $x$ and $x\in F(\phi)$, there exists a sequence $\{x_j\}_{j\geq 1}$ with $x_j \to x$ such that
\begin{equation}\label{eq:Wphi-xj-neg}
V_{\phi}(0,x_j,Q) \le -\frac{1}{j}\qquad \text{for all } j\ge 1.
\end{equation}
Using \eqref{eq:W-modulus-lambda} evaluated at $x_j$, we obtain
\begin{equation}\label{eq:W-triangle}
V_{\psi_j}(0,x_j,Q)
\le
V_{\phi}(0,x_j,Q)
+
\omega\!\bigl(|\psi_j-\phi|\bigr).
\end{equation}
\noindent Since $\omega\!\bigl(|\psi_j-\phi|\bigr) \to 0$ as $j\to\infty$, there exists a subsequence $\{\psi_{k_j}\}_{j\geq 1}$ such that
\begin{equation}\label{eq:omega-subseq}
\omega\!\bigl(|\psi_{k_j}-\phi|\bigr) \le \frac{1}{j}
\qquad \text{for all } j\ge 1.
\end{equation}

\noindent From~\eqref{eq:Wphi-xj-neg}, \eqref{eq:W-triangle}, and \eqref{eq:omega-subseq}, it follows that
\begin{equation}\label{eq:W-tri}
V_{\psi_{k_j}}(0,x_j,Q)
\le
-\frac{1}{j} + \frac{1}{j}
= 0.
\end{equation}

\noindent Equation~\eqref{eq:W-tri} shows that $V_{\psi_{k_j}}(0,x_j,Q)\le 0$ for all $j\ge 1$. This implies that $x_j \in F(\psi_{k_j})$ for all $j\ge 1$.
Define $x_{k_j} \coloneqq x_j$. Then $x_{k_j}\in F(\psi_{k_j})$ for all $j\ge 1$, and since $x_j\to x$, we also have \(x_{k_j}\) converges to  \(x\) as \(j\) goes to \(\infty. \)
Therefore, for every sequence $\psi_j\to \phi$ and every $x\in F(\phi)$, there exists a subsequence $\psi_{k_j}$ and a sequence
$x_{k_j}\in F(\psi_{k_j})$ such that $x_{k_j}\to x$. This proves that $F$ is lower hemicontinuous at $\phi$.
Since $\phi\in\Lambda$ was arbitrary, $F$ is lower hemicontinuous on $\Lambda$.

%%%%%%%%%%%%%%%%%%%%%%%%%%%%%%%%%%%%%%%%%%%%%%%%%%%%%%%%%%%%%%%%%%%%%%%%%%%%%%%%%%%%%%%%%%%%%%%%%%%%%%%%%%%%%%%%%%%%%%%%%%%%%%%%%%%%%%%%%%
\medskip
let $\lambda\in\Lambda$ be arbitrary and let $\{\lambda_k\}_{k\geq 1}\subset\Lambda$ satisfy $\lambda_k\to\lambda$.
Upper hemicontinuity of $F$ at $\lambda$ gives
\begin{equation}\label{eq:uhc}
\limsup_{k\to\infty} F(\lambda_k) \;\subseteq\; F(\lambda),
\end{equation}
while lower hemicontinuity of $F$ at $\lambda$ gives
\begin{equation}\label{eq:lhc}
F(\lambda) \;\subseteq\; \liminf_{k\to\infty} F(\lambda_k).
\end{equation}
Since $\liminf_{k\to\infty} F(\lambda_k)\subseteq \limsup_{k\to\infty} F(\lambda_k)$, it follows from
\eqref{eq:uhc}-\eqref{eq:lhc} that
\begin{equation}\label{eq:pk-equality}
\liminf_{k\to\infty} F(\lambda_k)
\;=\;
\limsup_{k\to\infty} F(\lambda_k)
\;=\;
F(\lambda).
\end{equation}

\noindent Under the assumptions stated in this paper, $F(\lambda_k)$ and $F(\lambda)$ are nonempty and compact for all $k$ (details omitted for brevity). Then the Painlev\'e-Kuratowski convergence in \eqref{eq:pk-equality} is equivalent to Hausdorff convergence, and hence
\begin{equation}\label{eq:hausdorff}
d_H\!\bigl(F(\lambda_k),F(\lambda)\bigr)\xrightarrow[k\to\infty]{}0.
\end{equation}

\end{proof}

\section{Proof of Theorem 3}\label{prop:penalty_function1a}

\begin{proof}
Let~\(V_{{\infty}}\) denote the value function with the same structure as \(V_\lambda\) in~\eqref{eq:lambda_monotonicity}, but with cost function~\(\ell_{\infty}\). Let \(\mathcal{RA}_{Q;{\infty}}\)~denote the associated soft-constrained reach-avoid set.

\noindent\textit{(a) Set inclusion.}  
As shown in the proof of Theorem~2 in Appendix~\ref{apx:Thm2}, for any \(\lambda \in [0,\infty)\), if \(\ell_{\lambda}(s,x) \leq \ell_{\infty}(s,x)\) for all \(s\in[0,T]\) and \(x\in\mathbb{C}_{1}\), then
\begin{equation}\label{eq:RA_inclusion}
\mathcal{RA}_{Q;\lambda} \supseteq \mathcal{RA}_{Q;{\infty}}.
\end{equation}

\noindent\textit{(b) Convergence in measure.}  
Fix \(Q \in [0,\infty)\). Let \(\{\ell_\lambda\}_{\lambda\in\Lambda}\in\mathcal{L}\) satisfy
\(\ell_\lambda(t,x)\to \ell_{\infty}(t,x)\) as \(\lambda\to+\infty\). Under the assumptions of Section~\ref{sec:problem}, \(V_{\lambda}(0,x,Q)\) converges pointwise on \(\mathbb{R}^n\) to \(V_{{\infty}}(0,x,Q)\)~\cite{mballo}.

Take any increasing sequence of cost functions from this family,
\((\ell_{\lambda_k})_{k\ge1}\), where \(\lambda_k\in\Lambda\) and \(\lambda_{k+1}\ge \lambda_k\). Let \(F:\Lambda \to 2^{\mathbb{K}}\) be the set-valued map defined in~\eqref{set_value_func}.

\noindent Since \(F(\lambda_{k+1})\subseteq F(\lambda_k)\) and \(\mu(\mathbb K)<\infty\), the continuity from above of the Lebesgue measure yields
\begin{equation}\label{eq:measure_limit}
\lim_{k\to\infty} \mu(F(\lambda_k)) \;=\; \mu\!\left(\bigcap_{k=1}^{\infty} F(\lambda_k)\right).
\end{equation}

Take any state $x\in \mathcal{RA}_{Q;{\infty}}$. By the definition of $\mathcal{RA}_{Q;{\infty}}$,
\begin{equation}
V_{{\infty}}(0,x,Q)\le 0.
\end{equation}

\noindent Following the argument in the proof of Theorem~2 (see \eqref{eq:lambda_monotonicity}), for each \(k\) we have
\begin{equation}
V_{\lambda_k}(0,y,Q)\le V_{{\infty}}(0,y,Q)
\qquad \text{for all } y \in \mathbb{R}^n .
\end{equation}

\noindent Therefore,
\begin{equation}
V_{\lambda_k}(0,x,Q)\le 0
\qquad \text{for every }k.
\end{equation}
Hence $x\in F(\lambda_k)$ for every $k$, and thus
\begin{equation}
x\in \bigcap_{k=1}^\infty F(\lambda_k).
\end{equation}
This proves
\begin{equation}\label{concl_1}
\mathcal{RA}_{Q;{\infty}}\subseteq \bigcap_{k=1}^\infty F(\lambda_k).
\end{equation}

\medskip
Now take $x\in \bigcap_{k=1}^\infty F(\lambda_k)$. Then $x\in F(\lambda_k)$ for every $k$.
By the definition of $F(\lambda_k)$,
\begin{equation}\label{W_k}
V_{\lambda_k}(0,x,Q)\le 0
\qquad \text{for every }k.
\end{equation}
Since \(V_{\lambda_k}(0,x,Q)\to V_{{\infty}}(0,x,Q)\) pointwise as \(k\to\infty\), letting \(k\to\infty\) in~\eqref{W_k} yields the inequality:
\begin{equation}
V_{{\infty}}(0,x,Q)\le 0.
\end{equation}
Therefore $x\in \mathcal{RA}_{Q;{\infty}}$. This proves
\begin{equation}\label{concl_2}
\bigcap_{k=1}^\infty F(\lambda_k)\subseteq \mathcal{RA}_{Q;{\infty}}.
\end{equation}

\medskip
\noindent Combining~\eqref{concl_1} and~\eqref{concl_2} yields
\begin{equation}\label{concl}
\bigcap_{k=1}^\infty F(\lambda_k) = \mathcal{RA}_{Q;{\infty}}.
\end{equation}

\noindent From~\eqref{eq:measure_limit} and~\eqref{concl}, we obtain
\begin{equation}\label{eq:final_limit}
\lim_{k\to\infty} \mu(F(\lambda_k)) = \mu(\mathcal{RA}_{Q;\,\ell_{\infty}}).
\end{equation}

\noindent Since \(\mathcal{RA}_{Q;{\infty}}\subseteq F(\lambda_k)\) for all \(k\), we have
\(F(\lambda_k)\triangle \mathcal{RA}_{Q;{\infty}}=F(\lambda_k)\setminus \mathcal{RA}_{Q;{\infty}}\), and therefore
\begin{equation}\label{eq:measure_gap_RA}
\mu\!\left(F(\lambda_k)\triangle \mathcal{RA}_{Q;{\infty}}\right)
=
\mu\!\left(F(\lambda_k)\right)
-
\mu\!\left(\mathcal{RA}_{Q;{\infty}}\right).
\end{equation}
Consequently,
\begin{equation}\label{eq:symmetric_difference}
\mu\!\left(\mathcal{RA}_{Q;\lambda_k}\triangle \mathcal{RA}_{Q;{\infty}}\right)
\;\xrightarrow[k\to\infty]{}\;0.
\end{equation}
\end{proof}

\bibliographystyle{IEEEtran}
\bibliography{refs}

@article{c1,
  author  = {Wabersich, K. P. and Taylor, A. J. and Choi, J. J. and Sreenath, K. and Tomlin, C. J. and Ames, A. D. and Zeilinger, M. N.},
  title   = {Data-driven safety filters: Hamilton--Jacobi reachability, control barrier functions, and predictive methods for uncertain systems},
  journal = {IEEE Control Systems Magazine},
  volume  = {43},
  number  = {5},
  pages   = {137--177},
  month   = oct,
  year    = {2023}
}

@article{c2,
  author  = {Ames, A. D. and Xu, X. and Grizzle, J. W. and Tabuada, P.},
  title   = {Control barrier function based quadratic programs for safety critical systems},
  journal = {IEEE Transactions on Automatic Control},
  volume  = {62},
  number  = {8},
  pages   = {3861--3876},
  month   = aug,
  year    = {2017}
}

@article{Tekles2017,
  author  = {Tekles, N. and Chongvisal, J. and Xargay, E. and Choe, R. and Talleur, D. A. and Hovakimyan, N. and Belcastro, C. M.},
  title = {Design of a flight envelope protection system for {NASA}'s transport class model},
  journal = {Journal of Guidance, Control, and Dynamics},
  volume  = {40},
  number  = {4},
  pages   = {863--877},
  year    = {2017}
}

@inproceedings{Tang2009,
  author    = {Tang, L. and Roemer, M. and Ge, J. and Crassidis, A. and Prasad, J. V. R. and Belcastro, C.},
  title     = {Methodologies for adaptive flight envelope estimation and protection},
  booktitle = {Proceedings of the AIAA Guidance, Navigation, and Control Conference},
  address   = {Chicago, IL, USA},
  month     = aug,
  year      = {2009}
}

@article{Yavrucuk2009,
  author  = {Yavrucuk, I. and Prasad, J. V. R. and Unnikrishnan, S.},
  title   = {Envelope protection for autonomous unmanned aerial vehicles},
  journal = {Journal of Guidance, Control, and Dynamics},
  volume  = {32},
  number  = {1},
  pages   = {248--261},
  year    = {2009}
}

@article{Akametalu,
  author  = {Akametalu, A. K. and Tomlin, C. J. and Chen, M.},
  title   = {Reachability-based forced landing system},
  journal = {Journal of Guidance, Control, and Dynamics},
  volume  = {41},
  number  = {12},
  pages   = {2529--2542},
  year  = {2018}  
}

@article{c3,
  author  = {Mitchell, I. M. and Bayen, A. M. and Tomlin, C. J.},
  title = {A time-dependent {Hamilton--Jacobi} formulation of reachable sets for continuous dynamic games},
  journal = {IEEE Transactions on Automatic Control},
  volume  = {50},
  number  = {7},
  pages   = {947--957},
  month   = jul,
  year    = {2005}
}

@article{MargellosLygeros2011,
  author  = {Margellos, K. and Lygeros, J.},
  title = {{Hamilton--Jacobi} formulation for reach--avoid differential games},
  journal = {IEEE Transactions on Automatic Control},
  volume  = {56},
  number  = {8},
  pages   = {1849--1861},
  month   = aug,
  year    = {2011}
}

@article{mballo,
  author  = {Mballo, Chams Eddine and Lee, D. and Tomlin, C. J.},
  title   = {{A Hamilton--Jacobi} Reachability Framework with Soft Constraints for Safety-Critical Systems},
  journal = {arXiv preprint arXiv:2510.24933},
  year    = {2025},
  note    = {Under review for IEEE Transactions on Automatic Control}
}

@article{bayen,
  author  = {Bayen, A. M. and Mitchell, I. M. and Oishi, M. M. K. and Tomlin, C. J.},
  title   = {Aircraft autolander safety analysis through optimal-control-based reach-set computation},
  journal = {Journal of Guidance, Control, and Dynamics},
  volume  = {30},
  number  = {1},
  pages   = {68--77},
  month   = jan,
  year    = {2007}
}

@article{Fisac,
  author  = {Fisac, J. F. and Akametalu, A. K. and Zeilinger, M. N. and Kaynama, S. and Gillula, J. and Tomlin, C. J.},
  title   = {A general safety framework for learning-based control in uncertain robotic systems},
  journal = {IEEE Transactions on Automatic Control},
  volume  = {64},
  number  = {7},
  pages   = {2737--2752},
  month   = jul,
  year    = {2019}
}

@inproceedings{Lombaerts2013,
  author    = {Lombaerts, T. and Schuet, S. and Wheeler, K. R. and Acosta, D. and Kaneshige, J.},
  title     = {Robust maneuvering envelope estimation based on reachability analysis in an optimal control formulation},
  booktitle = {Proceedings of the 3rd International Conference on Control and Fault-Tolerant Systems (SysTol)},
  address   = {Nice, France},
  month     = oct,
  year      = {2013},
  pages     = {1--8}
}

@article{Nabi2018,
  author  = {Nabi, H. N. and Lombaerts, T. and Zhang, Y. and van Kampen, E. and Chu, Q. P. and de Visser, C. C.},
  title   = {Effects of structural failure on the safe flight envelope of aircraft},
  journal = {Journal of Guidance, Control, and Dynamics},
  volume  = {41},
  number  = {6},
  pages   = {1257--1275},
  month   = jun,
  year    = {2018}
}

@article{Schuet2017,
  author  = {Schuet, S. and Lombaerts, T. and Acosta, D. and Kaneshige, J. and Wheeler, K. and Shish, K.},
  title   = {Autonomous flight envelope estimation for loss-of-control prevention},
  journal = {Journal of Guidance, Control, and Dynamics},
  volume  = {40},
  number  = {4},
  pages   = {847--862},
  year    = {2017}
}

@article{Hsu,
  author  = {Hsu, T. and Choi, J. J. and Amin, D. and Tomlin, C. J. and MacWherter, S. C. and Piedmonte, M.},
  title = {Towards flight envelope protection for the {NASA} tilt-wing {eVTOL} flight-mode transition using {Hamilton--Jacobi} reachability},
  journal = {Journal of the American Helicopter Society},
  volume  = {69},
  number  = {2},
  pages   = {1--18},
  month   = apr,
  year    = {2024}
}

@book{altman,
  author    = {Altman, Eitan},
  title     = {Constrained Markov Decision Processes},
  publisher = {CRC Press},
  address   = {Boca Raton, FL},
  year      = {1999}
}

@article{Russel2020,
  author  = {Russel, R. H. and Benosman, M. and van Baar, J.},
  title = {Robust constrained {MDPs}: Soft-constrained robust policy optimization under model uncertainty},
  journal = {arXiv preprint arXiv:2010.04870},
  year    = {2020}
}

@article{zelinger,
 author  = {Zeilinger, M. N. and Morari, M. and Jones, C. N.},
  title   = {Soft constrained model predictive control with robust stability guarantees},
  journal = {IEEE Transactions on Automatic Control},
  volume  = {59},
  number  = {5},
  pages   = {1190--1202},
  month   = may,
  year    = {2014}
}

@article{wabersich,
  author  = {Wabersich, K. P. and Krishnadas, R. and Zeilinger, M. N.},
  title   = {A soft constrained MPC formulation enabling learning from trajectories with constraint violations},
  journal = {IEEE Control Systems Letters},
  volume  = {6},
  pages   = {980--985},
  year    = {2022}
}

@inproceedings{lee,
  author    = {Lee, J. and Kim, J. and Ames, A. D.},
  title     = {Hierarchical relaxation of safety-critical controllers: Mitigating contradictory safety conditions with application to quadruped robots},
  booktitle = {Proceedings of the IEEE/RSJ International Conference on Intelligent Robots and Systems (IROS)},
  address   = {Detroit, MI, USA},
  month     = oct,
  year      = {2023},
  pages     = {2384--2391}
}

@inproceedings{xiao,
  author    = {Xiao, W. and Mehdipour, N. and Collin, A. and Bin-Nun, A. Y. and Frazzoli, E. and Tebbens, R. D. and Belta, C.},
  title     = {Rule-based optimal control for autonomous driving},
  booktitle = {Proceedings of the ACM/IEEE 12th International Conference on Cyber-Physical Systems (ICCPS)},
  address   = {Nashville, TN, USA},
  month     = may,
  year      = {2021},
  pages     = {143--154}
}

@inproceedings{Chow2017,
  author    = {Chow, Y. and Ghavamzadeh, M. and Janson, L. and Pavone, M.},
  title     = {Risk-Constrained Reinforcement Learning with Percentile Risk Criteria},
  booktitle = {Proceedings of the 34th International Conference on Machine Learning},
  series    = {Proceedings of Machine Learning Research},
  volume    = {70},
  address   = {Sydney, Australia},
  year      = {2017},
  pages     = {607--616}
}

@inproceedings{Tessler2019,
  author    = {Tessler, C. and Mankowitz, D. J. and Mannor, S.},
  title     = {Reward Constrained Policy Optimization},
  booktitle = {Proceedings of the AAAI Conference on Artificial Intelligence},
  address   = {Honolulu, HI, USA},
  year      = {2019},
  pages     = {5020--5027}
}

@inproceedings{mpc_soft_constraints1,
  author    = {Kerrigan, E. C. and Maciejowski, J. M.},
  title     = {Soft constraints and exact penalty functions in model predictive control},
  booktitle = {Proceedings of the UKACC International Conference on Control},
  address   = {Cambridge, U.K.},
  month     = sep,
  year      = {2000},
  pages     = {2319--2327}
}

@inproceedings{mpc_soft_constraints2,
  author    = {Oravec, J. and Bako{\v{s}}ov{\'a}, M.},
  title     = {Soft constraints in the robust MPC design via LMIs},
  booktitle = {Proceedings of the American Control Conference},
  address   = {Boston, MA, USA},
  month     = jul,
  year      = {2016},
  pages     = {3588--3593}
}

@article{mpc_soft_constraints5,
  author  = {Richards, A.},
  title   = {Fast model predictive control with soft constraints},
  journal = {European Journal of Control},
  volume  = {25},
  pages   = {51--59},
  month   = apr,
  year    = {2015}
}

@book{c16,
  author    = {Coddington, E. A. and Levinson, N.},
  title     = {Theory of Ordinary Differential Equations},
  publisher = {Krieger},
  address   = {Malabar, FL, USA},
  year      = {1984}
}

@article{crandall1983,
  author  = {Crandall, M. G. and Lions, P.-L.},
  title   =  {Viscosity solutions of {Hamilton--Jacobi} equations},
  journal = {Transactions of the American Mathematical Society},
  volume  = {277},
  number  = {1},
  pages   = {1--42},
  month   = may,
  year    = {1983}
}

@misc{FAR1990,
  author       = {{Federal Aviation Administration}},
  title        = {{14 CFR \S{}25.125: Landing}},
  howpublished = {Code of Federal Regulations, Title 14, Part 25},
  note         = {Airworthiness Standards: Transport Category Airplanes},
  year         = {2026}
}

@article{Evans,
  author  = {Evans, L. C. and Souganidis, P. E.},
  title = {Differential games and representation formulas for solutions of {Hamilton--Jacobi--Isaacs} equations},
  journal = {Indiana University Mathematics Journal},
  volume  = {33},
  number  = {5},
  pages   = {773--797},
  year    = {1984}
}

@misc{Toolbox_2,
  author       = {Chen, M. and Herbert, S. and Bansal, S. and Tomlin, C. J.},
  title        = {Optimal Control Helper Toolbox},
  howpublished = {\url{https://github.com/HJReachability/helperOC}},
  year         = {2023}
}

@article{Mitchell2008ToolboxLS,
  author  = {Mitchell, Ian M.},
  title   = {The Flexible, Extensible and Efficient Toolbox of Level Set Methods},
  journal = {Journal of Scientific Computing},
  volume  = {35},
  pages   = {300--329},
  year    = {2008}
}

@article{Chandrupatla1997,
  author  = {Chandrupatla, T. R.},
  title   = {A new hybrid quadratic/bisection algorithm for finding the zero of a nonlinear function without using derivatives},
  journal = {Advances in Engineering Software},
  volume  = {28},
  number  = {3},
  pages   = {145--149},
  year    = {1997}
}

@incollection{Walgemoed2005FlightEnvelope,
  author    = {Walgemoed, H.},
  title     = {Flight Envelope},
  booktitle = {Introduction to Flight Test Engineering},
  editor    = {Stoliker, F. N.},
  series    = {RTO AGARDograph 300, Flight Test Techniques Series},
  volume    = {14},
  chapter   = {12},
  publisher = {North Atlantic Treaty Organisation, Research and Technology Organisation},
  year      = {2005},
}

@article{Unnikrishnan2011ReactionaryEnvelopeProtection,
  author  = {Unnikrishnan, Suraj and Prasad, J. V. R. and Yavrucuk, Ilkay},
  title   = {Flight Evaluation of a Reactionary Envelope Protection System for UAVs},
  journal = {Journal of the American Helicopter Society},
  volume  = {56},
  number  = {1},
  year    = {2011},
  doi     = {10.4050/JAHS.56.012009}
}

@techreport{Whalley1994,
  author      = {Whalley, M. S.},
  title       = {A Piloted Simulation Investigation of Helicopter Limit Cueing},
  institution = {U.S. Army Aviation and Troop Command, Aeroflightdynamics Directorate},
  number      = {USAATCOM Technical Report 94-A-020},
  address     = {Moffett Field, CA},
  month       = oct,
  year        = {1994}
}

@techreport{WhalleyHindsonThiers2000,
  author      = {Whalley, M. S. and Hindson, B. and Thiers, G.},
  title       = {Comparison of Active Sidestick and Conventional Inceptors for Helicopter Flight Envelope Tactile Cueing},
  institution = {NASA Ames Research Center},
  address     = {Moffett Field, CA},
  month       = jan,
  year        = {2000}
}

@article{Falkena2011,
  author  = {Falkena, W. and Borst, C. and Chu, Q. P. and Mulder, J. A.},
  title   = {Investigation of Practical Flight Envelope Protection Systems for Small Aircraft},
  journal = {Journal of Guidance, Control, and Dynamics},
  volume  = {34},
  number  = {4},
  pages   = {976--988},
  month   = aug,
  year    = {2011},
  doi     = {10.2514/1.53000}
}

@inproceedings{Yavrucuk2002AdaptiveLimitMarginCueing,
  author    = {Yavrucuk, Ilkay and Prasad, J. V. R.},
  title     = {Adaptive Limit Margin Prediction and Control Cueing for Carefree Maneuvering of VTOL Aircraft},
  booktitle = {Proceedings of the AHS Flight Controls and Crew System Design Technical Specialists' Meeting},
  address   = {Philadelphia, PA},
  month     = oct,
  year      = {2002}
}

@techreport{Lotterio2022AdvancedFlightControls,
  author      = {Lotterio, Marco},
  title       = {Develop a Method of Compliance to Support Certification of Advanced Flight Controls in General Aviation and Hybrid Vehicles},
  institution = {National Test Pilot School},
  number      = {DOT/FAA/TC-21/6},
  address     = {Atlantic City International Airport, NJ},
  month       = feb,
  year        = {2022},
  doi         = {10.21949/1524486}
}

@inproceedings{Klyde2020HQTE,
  author    = {Klyde, David H. and Schulze, P. Chase and Mitchell, David G. and Sizoo, David and Schaller, Ross and McGuire, Robert},
  title     = {Mission Task Element Development Process: An Approach to {FAA} Handling Qualities Certification},
  booktitle = {AIAA Aviation Forum},
  year      = {2020},
  doi       = {10.2514/6.2020-3285}
}

\begin{IEEEbiography}[{\includegraphics[width=1in,height=1.25in,clip,keepaspectratio]{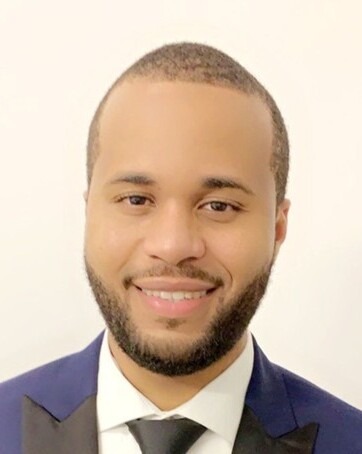}}]{Chams Eddine Mballo}
received the Ph.D. degree in aerospace engineering (2022) and the M.S. degree in mathematics (2021) from the Georgia Institute of Technology, Atlanta. He is currently a Postdoctoral Fellow at the University of California, Berkeley. His research focuses on safety-critical control, reachability analysis, human–machine interaction, and sustainable aviation for next-generation electric aircraft.
\end{IEEEbiography}

\begin{IEEEbiography}[{\includegraphics[width=1in,height=1.25in,clip,keepaspectratio]{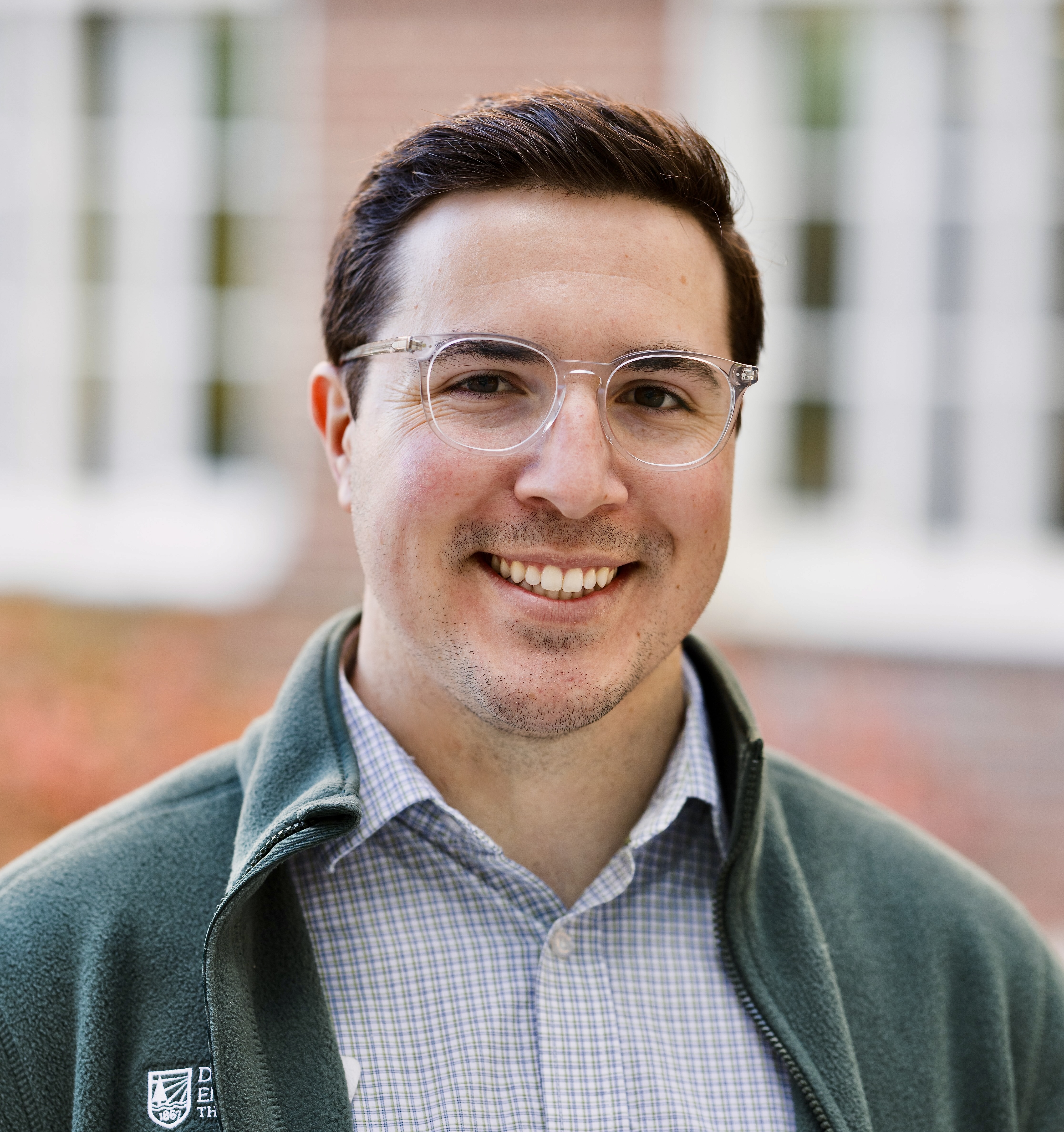}}]{Bryce L. Ferguson}
is an Assistant Professor in the Thayer School of Engineering at Dartmouth College, Hanover, NH, USA, where he leads the MADCAT Lab. He received the A.A. degree in mathematics from Santa Rosa Junior College, Santa Rosa, CA, USA, in 2016, and the B.S. and M.S. degrees in electrical engineering from the University of California, Santa Barbara, CA, USA, in 2018 and 2020, respectively. He received the Ph.D. degree in electrical and computer engineering from the University of California, Santa Barbara in 2024. He was a Postdoctoral Researcher at the University of California, Berkeley in 2024-2025. His research interests include game-theoretic and control-theoretic methods for describing and controlling both societal and engineered multiagent systems.
\end{IEEEbiography}

\begin{IEEEbiography}[{\includegraphics[width=1in,height=1.25in,clip,keepaspectratio]{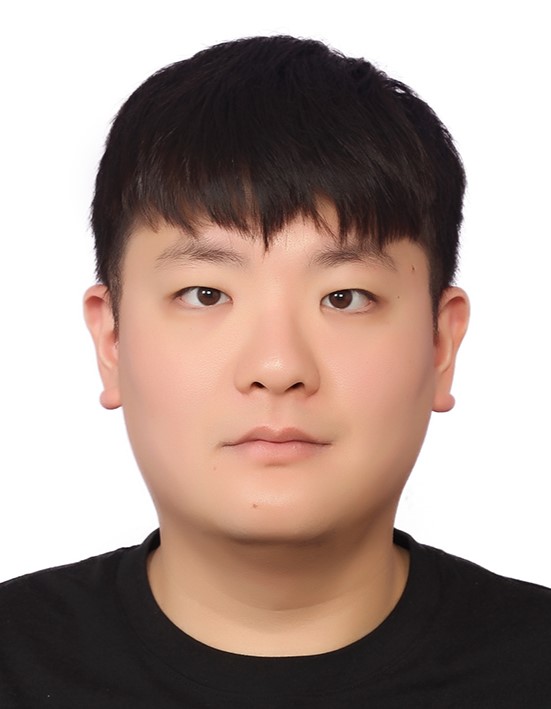}}]{Inkyu Jang}
(Graduate Student Member, IEEE) received the B.S. degree in mechanical engineering from Seoul National University, Seoul, Korea, in 2020. He is currently a Ph.D. candidate in Aerospace Engineering at Seoul National University, Seoul, Korea. His research interests include safety-critical control and stochastic control and their connection with machine learning, with applications to robotics.
\end{IEEEbiography}

\begin{IEEEbiography}[{\includegraphics[width=1in,height=1.25in,clip,keepaspectratio]{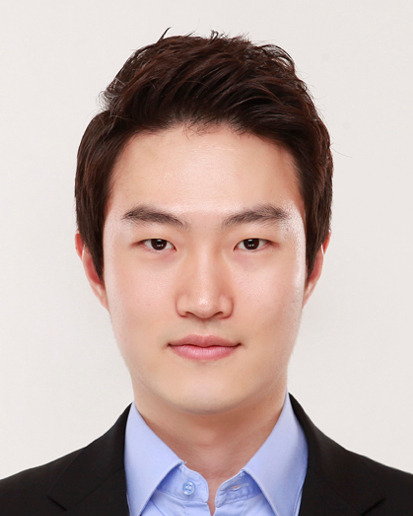}}]{Donggun Lee}
received the Ph.D. degree in mechanical engineering from the University of California, Berkeley, CA, USA. He was a Postdoctoral Associate with the Massachusetts Institute of Technology (MIT), Cambridge, MA, USA. He is currently an Assistant Professor with the Department of Mechanical and Aerospace Engineering, North Carolina State University (NCSU), Raleigh, NC, USA. His research interests include control theory and machine learning, with applications to robotics and vehicles.
\end{IEEEbiography}

\begin{IEEEbiography}[{\includegraphics[width=1in,height=1.25in,clip,keepaspectratio]{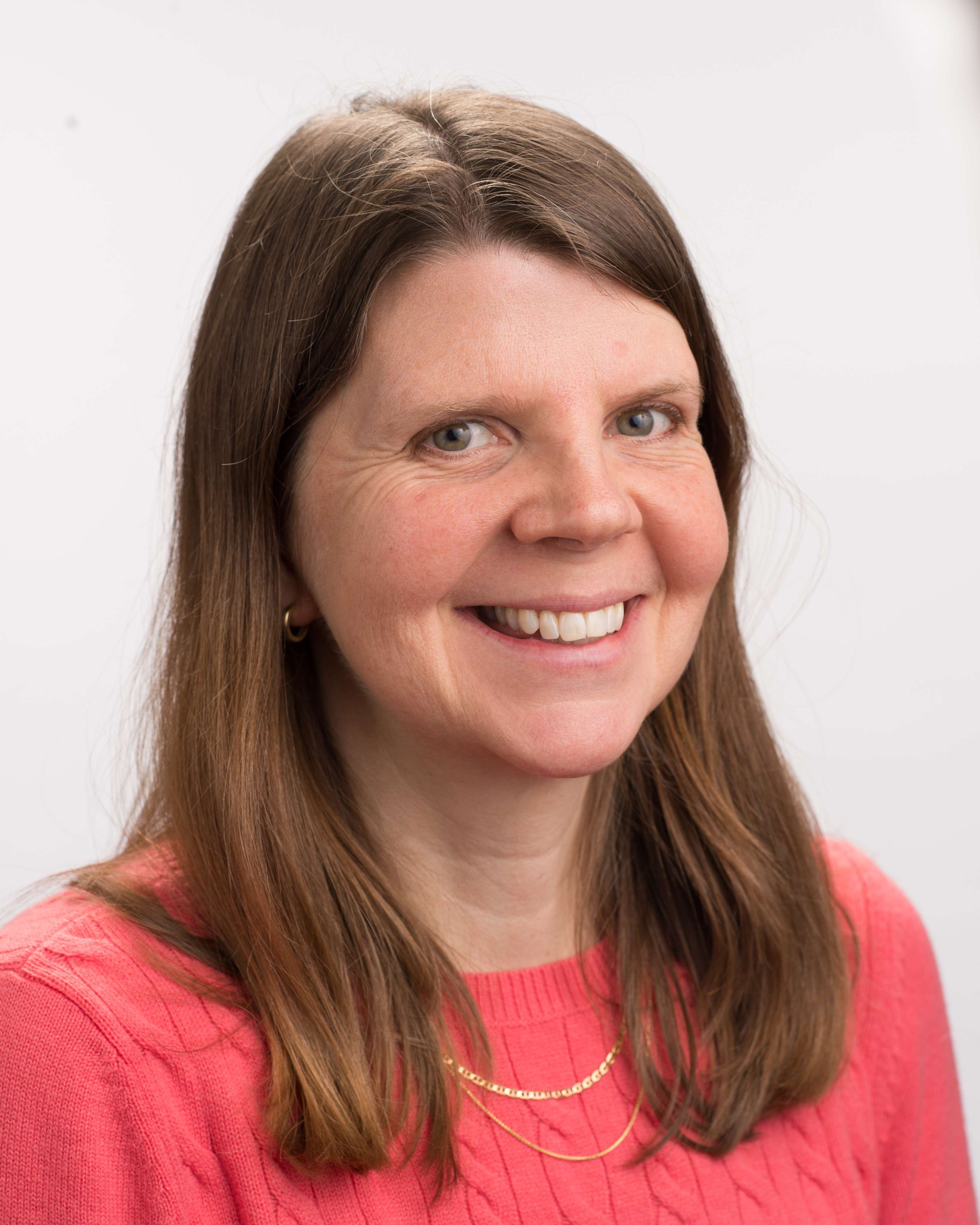}}]{Claire J. Tomlin}
(Fellow, IEEE) received the Ph.D. degree in electrical engineering and computer science from the University of California, Berkeley, CA, USA, in 1998. She is the James and Katherine Lau Professor of Engineering and the Professor and Chair with the Department of Electrical Engineering and Computer Sciences at UC Berkeley, Berkeley, CA, USA. From 1998 to 2007, she was an Assistant, Associate, and Full Professor in Aeronautics and Astronautics at Stanford University, Stanford, CA. In 2005, she joined UC Berkeley. Her research interests include control theory and hybrid systems, with applications to air traffic management, UAV systems, energy, robotics, and systems biology.

Prof. Tomlin is a MacArthur Foundation Fellow (2006). She was the recipient of the IEEE Transportation Technologies Award in 2017. In 2019, she was elected to the National Academy of Engineering and the American Academy of Arts and Sciences.
\end{IEEEbiography}

\vfill

\end{document}